%% file: chapter.tex
\RequirePackage{snapshot}
\newif\ifsubsections
\subsectionstrue
\documentclass[11pt]{article}
\newlength{\defbaselineskip}
\usepackage[sc]{mathpazo}          
\usepackage{eulervm}               
\usepackage[scaled=0.86]{berasans} 
\usepackage[scaled=1]{inconsolata} 
\usepackage[T1]{fontenc}
\usepackage[%
	protrusion=true,
	expansion=false,
	auto=false
	]{microtype}
\usepackage{xcolor}
\usepackage{graphicx}
\graphicspath{{./figures/}}
	\definecolor{linkred}{rgb}{0.0,0.0,0.0}
	\definecolor{linkblue}{rgb}{0,0.0,0.0}
\usepackage{url}
\usepackage{amsfonts,amsmath,amssymb}
\usepackage{amsthm}
\usepackage{algorithm2e,algorithmic,framed}
\PassOptionsToPackage{hyphens}{url}
\usepackage[
    setpagesize=false,
    pagebackref,
	pdfpagelabels=false,
    pdfstartview={FitH 1000},
    bookmarksnumbered=false,
    linktoc=all,
    colorlinks=true,
    anchorcolor=black,
    menucolor=black,
    runcolor=black,
    filecolor=black,
    linkcolor=linkblue,
	citecolor=linkblue,
	urlcolor=linkred,
]{hyperref}
\usepackage[backrefs,msc-links,nobysame]{amsrefs}
\theoremstyle{plain}

\theoremstyle{definition}

\newtheorem{definition}[equation]{Definition}
\newtheorem{Remark}[equation]{Remark}{\rmfamily}

\newtheorem{lemma}[equation]{Lemma}
\newtheorem{theorem}[equation]{Theorem}

\newcommand{\Probab}[1]{\mbox{}{\bf{Pr}}\left[#1\right]} 
\newcommand{\Expect}[1]{\mbox{}{\bf{E}}\left[#1\right]} 
\newcommand{\Var}[1]{\mbox{}{\bf{Var}}\left[#1\right]} 
\newcommand{\Trace }[1]{\mbox{}{\bf{Tr}}\left(#1\right)} 
\newcommand{\ONorm }[1]{\mbox{}\|#1\|_{1}} 
\newcommand{\FNorm }[1]{\mbox{}\|#1\|_F  } 
\newcommand{\FNormS}[1]{\mbox{}\|#1\|_F^2} 
\newcommand{\TNorm}[1]{\mbox{}\|#1\|_2} 
\newcommand{\TNormS}[1]{\mbox{}\|#1\|_2^2} 
\newcommand{\PNorm }[1]{\mbox{}\|#1\|_p  } 
\newcommand{\VINorm }[1]{\mbox{}\|#1\|_{\infty}  } 
\newcommand{\abs }[1]{\left|#1\right|} 
\newcommand{\bigO}[1]{O(#1)} 
\newcommand{\rank}[1]{\mbox{rank}(#1)} 
\newcommand{\range}[1]{\mbox{range}(#1)} 
\newcommand{\rnull}[1]{\mbox{null}(#1)} 
\newcommand{\nnz}[1]{\mbox{nnz}(#1)} 

\newcommand{\bA}{\mathbf{A}}

\newcommand{\bD}{\boldsymbol{D}}
\newcommand{\bd}{\boldsymbol{d}}

\newcommand{\bV}{\boldsymbol{V}}
\newcommand{\bv}{\boldsymbol{v}}
\newcommand{\be}{\boldsymbol{e}}

\newcommand{\bU}{\boldsymbol{U}}
\newcommand{\bu}{\boldsymbol{u}}
\newcommand{\bQ}{\boldsymbol{Q}}
\newcommand{\bW}{\boldsymbol{W}}
\newcommand{\bP}{\boldsymbol{P}}
\newcommand{\bX}{\boldsymbol{X}}
\newcommand{\bx}{\boldsymbol{x}}

\newcommand{\bY}{\boldsymbol{Y}}
\newcommand{\by}{\boldsymbol{y}}
\newcommand{\bZ}{\boldsymbol{Z}}
\newcommand{\bz}{\boldsymbol{z}}
\newcommand{\bI}{\boldsymbol{I}}
\newcommand{\bB}{\boldsymbol{B}}
\newcommand{\bvb}{\boldsymbol{b}}
\newcommand{\bC}{\boldsymbol{C}}
\newcommand{\bH}{\boldsymbol{H}}

\newcommand{\bR}{\boldsymbol{R}}

\newcommand{\bS}{\boldsymbol{S}}

\newcommand{\bSigma}{\boldsymbol{\Sigma}}
\newcommand{\bPhi}{\boldsymbol{\Phi}}

\newcommand{\bzero}{\boldsymbol{0}}
\newcommand{\bone}{\boldsymbol{1}}
\newcommand{\Rs}[1]{\mathbb{R}^{#1}} 

\begin{document}
%
%
%
%
%
%
\title{Lectures on Randomized Numerical Linear Algebra%
\footnote{This chapter is based on lectures from the 2016 Park City Mathematics Institute summer school on \emph{The Mathematics of Data}, and it appears as a chapter~\cite{pcmi-chapter-randnla} in the edited volume of lectures from that summer school.}
}
%
%
\author{
Petros Drineas
\thanks{
Purdue University, Computer Science Department, 305 N University Street, West Lafayette, IN 47906.
Email: \texttt{pdrineas@purdue.edu}.
}
\and
Michael W. Mahoney
\thanks{
University of California at Berkeley, ICSI and Department of Statistics, 367 Evans Hall, Berkeley, CA 94720.
Email: \texttt{mmahoney@stat.berkeley.edu}.
}
}
%
%
%
%
\date{}
\maketitle
%
%
\tableofcontents

\section{Introduction}\label{sxn:intro}

Matrices are ubiquitous in computer science, statistics, and applied mathematics. An $m \times n$ matrix can encode information about $m$ objects (each described by $n$ features), or the behavior of a discretized differential operator on a finite element mesh; an $n \times n$ positive-definite matrix can encode the correlations between all pairs of $n$ objects, or the edge-connectivity between all pairs of $n$ nodes in a social network; and so on. Motivated largely by technological developments that generate extremely large scientific and Internet data sets, recent years have witnessed exciting developments in the theory and practice of matrix algorithms. Particularly remarkable is the use of \textit{randomization}---typically assumed to be a property of the input data due to, e.g., noise in the data generation mechanisms---as an algorithmic or computational resource for the development of improved algorithms for fundamental matrix problems such as matrix multiplication, least-squares (LS) approximation, low-rank matrix approximation,~etc.

Randomized Numerical Linear Algebra (RandNLA) is an interdisciplinary research area that exploits randomization as a computational resource to develop improved algorithms for large-scale linear algebra problems. From a foundational perspective, RandNLA has its roots in theoretical
computer science (TCS), with deep connections to mathematics (convex analysis, probability theory, metric embedding theory) and applied mathematics (scientific computing, signal processing, numerical linear algebra). From an applied perspective, RandNLA is a vital new tool for
machine learning, statistics, and data analysis. Well-engineered implementations have already outperformed highly-optimized software libraries for ubiquitous problems such as least-squares regression, with good scalability in parallel and distributed environments. Moreover, RandNLA promises a sound algorithmic and statistical foundation for modern large-scale data analysis.

This chapter serves as a self-contained, gentle introduction to three fundamental RandNLA algorithms: randomized matrix multiplication, randomized least-squares solvers, and a randomized algorithm to compute a low-rank approximation to a matrix.
As such, this chapter has strong connections with many areas of applied mathematics, and in particular it has strong connections with several other chapters in this volume. Most notably, this includes that of G.~Martinsson, who uses these methods to develop improved low-rank matrix approximation solvers~\cite{pcmi-chapter-martinsson}; R.~Vershynin, who develops probabilistic tools that are used in the analysis of RandNLA algorithms~\cite{pcmi-chapter-vershynin}; J.~Duchi, who uses stochastic and randomized methods in a complementary manner for more general optimization problems~\cite{pcmi-chapter-duchi}; and M.~Maggioni, who uses these methods as building blocks for more complex multiscale methods~\cite{pcmi-chapter-maggioni}.

We start this chapter with a review of basic linear algebraic facts in Section~\ref{sxn:introla}; we review basic facts from discrete probability in Section~\ref{sxn:dp}; we present a randomized algorithm for matrix multiplication in Section~\ref{chapter:MM}; we present a randomized algorithm for least-squares regression problems in Section~\ref{sxn:main:regression}; and finally we present a randomized algorithm for low-rank approximation in Section~\ref{sxn:main:lowrank}.
We conclude this introduction by noting that~\cite{Drineas2016,Mah-mat-rev_BOOK} might also be of interest to a reader who wants to go through other introductory texts on RandNLA.

\section{Linear Algebra}\label{sxn:introla}
 
In this section, we present a brief overview of basic linear algebraic facts and notation that will be useful in this chapter. We assume basic familiarity with linear algebra (e.g., inner/outer products of vectors, basic matrix operations such as addition, scalar multiplication, transposition, upper/lower triangular matrices, matrix-vector products, matrix multiplication, matrix trace, etc.).

\subsection{Basics.}\label{sxn:labasics}
We will entirely focus on matrices and vectors over the \textit{reals}. We will use the notation $\bx \in \Rs{n}$ to denote an $n$-dimensional vector: notice the use of bold latin \textit{lowercase} letters for vectors. Vectors will always be assumed to be column vectors, unless explicitly noted otherwise. The vector of all zeros will be denoted as $\bzero$, while the vector of all ones will be denoted as $\bone$; dimensions will be implied from context or explicitly included as a subscript.

We will use bold latin \textit{uppercase} letters for matrices, e.g., $\bA \in \mathbb{R}^{m \times n}$ denotes an $m \times n$ matrix $\bA$. We will use the notation $\bA_{i*}$ to denote the $i$-th row of $\bA$ as a row vector and $\bA_{*i}$ to denote the $i$-th column of $\bA$ as a column vector. The (square) identity matrix will be denoted as $\bI_n$ where $n$ denotes the number of rows and columns. Finally, we use $\be_i$ to denote the $i$-th column of $\bI_n$, i.e., the $i$-th \textit{canonical} vector.

\noindent\textbf{Matrix Inverse.} A matrix $\bA\in\Rs{n\times n}$ is nonsingular or invertible if there exists a matrix $\bA^{-1} \in \Rs{n \times n}$ such that
\[\bA\bA^{-1}=\bI_{n\times n}=\bA^{-1}\bA.\]
The inverse exists when all the columns (or all the rows) of $\bA$ are linearly independent. 
In other words, there does not exist a non-zero vector  $\bx \in \Rs{n}$ such that $\bA\bx = \bzero$. 
Standard properties of the inverse include: $(\bA^{-1})^\top = (\bA^{\top})^{-1} = \bA^{-\top}$ and $(\bA\bB)^{-1} = \bB^{-1}\bA^{-1}$.

\noindent \textbf{Orthogonal matrix.} A matrix $\bA\in\Rs{n\times n}$ is orthogonal if $\bA^\top = \bA^{-1}$.
Equivalently, for all $i$ and $j$ between one and $n$,
\[ \bA_{*i}^\top\bA_{*j} =
\begin{cases}
0,       &  \text{if }i\neq j\\
1,       &  \text{if } i = j
\end{cases}.\]
The same property holds for the rows of $\bA$. In words, the columns (rows) of $\bA$ are pairwise orthogonal and normal vectors.

\noindent\textbf{QR Decomposition.}
Any matrix $\bA \in \Rs{n\times n}$ can be decomposed into the product of an orthogonal matrix and an upper triangular matrix as:
\[\bA = \bQ\bR,\]
where $\bQ \in \Rs{n\times n}$ is an orthogonal matrix and $\bR \in \Rs{n\times n}$ is an upper triangular matrix.
The QR decomposition is useful in solving systems of linear equations, has computational complexity $\bigO{n^3}$, and is numerically stable.
To solve the linear system $\bA \bx = \bvb$ using the QR decomposition we first premultiply both sides by $\bQ^\top$, thus getting $\bQ^\top\bQ\bR\bx =\bR\bx= \bQ^\top\bvb$.
Then, we solve $\bR\bx = \bQ^\top\bvb$ using backward substitution~\cite{GVL96}.

\subsection{Norms.}
Norms are used to measure the size or mass of a matrix or, relatedly, the length of a vector. They are functions that map an object from $\Rs{m \times n}$ (or $\Rs{n}$) to $\Rs{}$. Formally:
\begin{definition}Any function, $\|\cdot\|$: $\Rs{m\times n} \rightarrow \Rs{}$ that satisfies the following properties is called a {\bf norm}:
\begin{enumerate}
	\item Non-negativity: $\|\bA\|\geq 0$; $\|\bA\|= 0$ if and only if $\bA = \bzero$.
	\item Triangle inequality: $\|\bA+\bB\| \leq \|\bA\|+\|\bB\|$.
	\item Scalar multiplication: $\|\alpha \bA\| = |\alpha|\|\bA\|$, for all $\alpha\in\Rs{}$.
\end{enumerate}\end{definition}
\noindent The following properties are easy to prove for any norm: $\|-\bA\| = \|\bA\|$ and \[|\|\bA\|-\|\bB\||\leq \|\bA-\bB\|.\] The latter property is known as the reverse triangle inequality.

\subsection{Vector norms.} Given $\bx \in \Rs{n}$ and an integer $p\geq 1$, we define the vector $p$-norm as:
\[\PNorm{\bx} = \left(\sum_{i=1}^n|x_i|^p\right)^{1/p}.\]
The most common vector $p$-norms are:
\begin{itemize}
	\item One norm: $\ONorm{\bx} = \sum_{i=1}^n|x_i|$.
	\item Euclidean (two) norm: $\TNorm{\bx} = \sqrt{\sum_{i=1}^n|x_i|^2} = \sqrt{\bx^\top\bx}$.
	\item Infinity (max) norm: $\VINorm{\bx} = \max_{1\leq i \leq n}|x_i|$.
\end{itemize}
Given $\bx,\by \in \Rs{n}$ we can bound the inner product $\bx^\top\by = \sum_{i=1}^n x_iy_i$ using $p$-norms. The Cauchy-Schwartz inequality states that:
\[|\bx^\top\by| \leq\TNorm{\bx}\TNorm{\by}.\]
In words, it gives an upper bound for the inner product of two vectors in terms of the Euclidean norm of the two vectors. H\"{o}lder's inequality states that
\[|\bx^\top\by|\leq\ONorm{\bx}\VINorm{\by}\quad \text{and}\quad |\bx^\top\by|\leq\VINorm{\bx}\ONorm{\by}.\]
The following inequalities between common vector $p$-norms are easy to prove:
\begin{align*}
\VINorm{\bx}\leq   \ONorm{\bx} &\leq n\VINorm{\bx},\\
\TNorm{\bx}\leq         \ONorm{\bx} &\leq \sqrt{n}\TNorm{\bx},\\
\VINorm{\bx}\leq    \TNorm{\bx} &\leq \sqrt{n}\VINorm{\bx}.
\end{align*}
Also, $\TNorm{\bx}^2 = \bx^T\bx$.
We can now define the notion of orthogonality for a pair of vectors and state the Pythagorean theorem.
\begin{theorem}\label{thm:pythagoras}
Two vectors $\bx,\by\in \Rs{n}$ are orthogonal, i.e., $\bx^\top\by = 0$, if and only~if
\[\TNormS{\bx\pm\by} =\TNormS{\bx}+\TNormS{\by}.\]
\end{theorem}
Theorem~\ref{thm:pythagoras} is also known as the Pythagorean Theorem. Another interesting property of the Euclidean norm is that it does not change after pre(post)-multiplication by a matrix with orthonormal columns (rows).
\begin{theorem}
	Given a vector $\bx\in\Rs{n}$ and a matrix $\bV \in \Rs{m\times n}$ with $m \ge n$ and $\bV^\top\bV = \bI_n$:
	\[\TNorm{\bV\bx} =\TNorm{\bx}\quad \mbox{and} \quad \TNorm{\bx^T\bV^T} =\TNorm{\bx}.\]
\end{theorem}

\subsection{Induced matrix norms.} Given a matrix $\bA\in \Rs{m\times n}$ and an integer $p\geq 1$ we define the matrix $p$-norm as:
	\[\PNorm{\bA} = \max_{\bx\neq 0}{\frac{\PNorm{\bA\bx}}{\PNorm {\bx}}} = \max_{\PNorm {\bx} = 1}{\PNorm{\bA\bx}}.\]
	The most frequent matrix $p$-norms are:
	\begin{itemize}
		\item One norm: the maximum absolute column sum,
		\[\ONorm{\bA} = \max_{1\leq j\leq n}\sum_{i=1}^m{|\bA_{ij}|} = \max_{1\leq j \leq n}\ONorm{\bA \be_j}.\]
		\item Infinity norm: the maximum absolute row sum,
			\[\VINorm{\bA} = \max_{1\leq i\leq m}\sum_{j=1}^n{|\bA_{ij}|} = \max_{1\leq j \leq m}\ONorm{\bA^\top \be_i}.\]
			\item Two (or spectral) norm :
			\[\TNorm{\bA} = \max_{\TNorm{\bx} = 1}\|\bA\bx\|_2 =  \max_{\TNorm{\bx} = 1} \sqrt{\bx^\top\bA^\top\bA\bx } .\]
	\end{itemize}
This family of norms is named ``induced'' because they are realized by a non-zero vector $\bx$ that varies depending on $\bA$ and $p$. Thus, there exists a unit norm vector (unit norm in the $p$-norm) $\bx$ such that
$\PNorm{\bA} = \PNorm{\bA\bx}.$
The induced matrix $p$-norms follow the submultiplicativity laws:
\[\PNorm{\bA\bx} \leq \PNorm{\bA}\PNorm{\bx}\qquad \mbox{and} \qquad \PNorm{\bA\bB}\leq \PNorm{\bA}\PNorm{\bB}.\]
%
Furthermore, matrix $p$-norms are invariant to permutations:
$\PNorm{\bP\bA\bQ} = \PNorm{\bA},$
where $\bP$ and $\bQ$ are permutation matrices of appropriate dimensions. Also, if we consider the matrix with permuted rows and columns
\[\bP\bA\bQ = \begin{pmatrix}\bB & \bA_{12}\\ \bA_{21} & \bA_{22} \end{pmatrix}  ,  \]
then the norm of the submatrix is related to the norm of the full unpermuted matrix as follows:
$\PNorm{\bB}\leq \PNorm{\bA}.$
The following relationships between matrix $p$-norms are relatively easy to prove. Given a matrix $\bA\in \Rs{m\times n}$,
\begin{align*}
\frac{1}{\sqrt{n}}\VINorm{\bA}  \leq  \TNorm{\bA}  \leq\sqrt{m}\VINorm{\bA}, \\
\frac{1}{\sqrt{m}}\ONorm{\bA}  \leq  \TNorm{\bA}  \leq\sqrt{n}\ONorm{\bA}.
\end{align*}
It is also the case that
$\ONorm{\bA^\top} =  \VINorm{\bA}$ and $\VINorm{\bA^\top}  =\ONorm{\bA}.$
While transposition affects the infinity and one norm of a matrix, it does not affect the two norm, i.e.,
$\TNorm{\bA^\top} = \TNorm{\bA}$.
Also, the matrix two-norm is not affected by pre(post)- multiplication with matrices whose columns (rows) are orthonormal vectors:
$\TNorm{\bU\bA\bV^\top} = \TNorm{\bA},$
where $\bU$ and $\bV$ are orthonormal matrices ($\bU^T\bU=\bI$ and $\bV^T\bV=\bI$) of appropriate dimensions.

\subsection{The Frobenius norm.}\label{sxn:review:FNorm}
The Frobenius norm is not an induced norm, as it belongs to the family of Schatten norms (to be discussed in Section~\ref{sxn:schatten}). 
\begin{definition}
	Given a matrix $\bA \in \Rs{m\times n}$, we define the Frobenius norm as:
	\[\FNorm{\bA} = \sqrt{\sum_{j=1}^n\sum_{i=1}^m{\bA_{ij}^2}} = \sqrt{\Trace{\bA^\top\bA}},\]
	where $\Trace{\cdot}$ denotes the matrix trace (where, recall, the trace of a square matrix is  defined to be the sum of the elements on the main diagonal).
\end{definition}
Informally, the Frobenius norm measures the variance or variability (which can be given an interpretation of size or mass) of a matrix.
Given a vector $\bx \in \Rs{n}$, its Frobenius norm is equal to its Euclidean norm, i.e.,
$\FNorm{\bx} = \TNorm{\bx}$. Transposition of a matrix $\bA\in \Rs{m\times n}$ does not affect its Frobenius norm, i.e., $\FNorm{\bA}=\FNorm{\bA^\top}$.
 Similar to the two norm, the Frobenius norm does not change under permutations or pre(post)- multiplication with a matrix with orthonormal columns (rows):
 \[\FNorm{\bU\bA\bV^T}= \FNorm{\bA},\]
where $\bU$ and $\bV$ are orthonormal matrices ($\bU^T\bU=\bI$ and $\bV^T\bV=\bI$) of appropriate dimensions.
The two and the Frobenius norm can be related by:
\[\TNorm{\bA} \leq \FNorm{\bA} \leq \sqrt{\rm{rank}(\bA)}\TNorm{\bA}\leq\sqrt{\min{\{m,n\}}}\TNorm{\bA}.\]
The Frobenius norm satisfies the so-called strong sub-multiplicativity property, namely:
\[\FNorm{\bA\bB}\leq \TNorm{\bA}\FNorm{\bB} \quad \mbox{and} \quad \FNorm{\bA\bB}\leq \FNorm{\bA}\TNorm{\bB}.\]
Given $\bx \in \Rs{m}$ and $\by \in \Rs{n}$, the Frobenius norm of their outer product is equal to the product of the Euclidean norms of the two vectors forming the outer product:
\[\FNorm{\bx\by^\top} =\TNorm{\bx}\TNorm{\by}.\]
Finally, we state a matrix version of the Pythagorean theorem.
\begin{lemma}[Matrix Pythagoras]\label{l_pyth}
Let $\bA,\bB\in\mathbb{R}^{m\times n}$. If $\bA^T\bB=\bzero$ then
\[\|\bA+\bB\|_F^2=\|\bA\|_F^2+\|\bB\|_F^2.\]
\end{lemma}

\subsection{The Singular Value Decomposition.}\label{sxn:review:SVD}
The Singular Value Decomposition (SVD) is the most important matrix decomposition and exists for every matrix.
\begin{definition}
	Given a matrix $\bA\in\Rs{m\times n}$, we define its full SVD as:
	\[\bA = \bU \bSigma \bV^T
	= \sum_{i=1}^{\min\{m,n\}}\sigma_i\bu_i\bv_i^\top
	,\]
	where $\bU\in\Rs{m\times m}$ and $\bV\in\Rs{n\times n}$ are orthogonal matrices that contain the left and right singular vectors of $\bA$, respectively, and $\bSigma \in \Rs{m\times n}$ is a diagonal matrix, with the singular values of $\bA$ in decreasing order on the diagonal.
\end{definition}
 We will often use $\bu_i$ (respectively, $\bv_j$), $i=1,\ldots, m$ (respectively, $j=1,\ldots, n$) to denote the columns of the matrix $\bU$ (respectively, $\bV$). Similarly, we will use $\sigma_i$, $i=1,\ldots, \min\{m,n\}$ to denote the singular values:
\[\sigma_1\geq \sigma_2\geq\cdots\geq \sigma_{\min\{m,n\}}\geq 0.\]
The singular values of $\bA$ are non-negative and their number is equal to $\min\{m,n\}$. The number of non-zero singular values of $\bA$ is equal to the rank of $\bA$. Due to orthonormal invariance, we get:
\[\bSigma_{\bP\bA\bQ^T} = \bSigma_{\bA},\]
where $\bP$ and $\bQ$ are orthonormal matrices ($\bP^T\bP=\bI$ and $\bQ^T\bQ=\bI$) of appropriate dimensions.
In words, the singular values of $\bP\bA\bQ$ are the same as the singular values of $\bA$.

The following inequalities involving the singular values of the matrices $\bA$ and $\bB$ are important. First, if both $\bA$ and $\bB$ are in $\mathbb{R}^{m \times n}$, for all $i=1,\ldots, \min\{m,n\}$,
\begin{align}
\label{eqn:svineq1}\abs{\sigma_i(\bA)-\sigma_i(\bB)}\leq\TNorm{\bA-\bB}.
\end{align}
Second, if $\bA \in \mathbb{R}^{p \times m}$ and $\bB \in \mathbb{R}^{m \times n}$, for all $i=1,\ldots, \min\{m,n\}$,
\begin{align}
\label{eqn:svineq2}\sigma_i(\bA\bB)\leq \sigma_1(\bA)\sigma_i(\bB),
\end{align}
where, recall, $ \sigma_1(\bA) = \TNorm{\bA} $.
We are often interested in keeping only the non-zero singular values and the corresponding left and right singular vectors of a matrix $\bA$. Given a matrix $\bA\in\Rs{m\times n}$ with $\rank{\bA} = \rho$, its thin SVD can be defined as follows.
\begin{definition}
Given a matrix $\bA\in\Rs{m\times n}$ of rank $\rho \leq \min\{m,n\}$, we define its thin SVD as:
\[\bA = \underbrace{\bU}_{m\times \rho}\underbrace{\bSigma}_{\rho\times \rho}\underbrace{\bV^\top}_{\rho\times n} = \sum_{i=1}^\rho\sigma_i\bu_i\bv_i^\top,\]
where $\bU\in\Rs{m\times \rho}$ and $\bV\in\Rs{n\times \rho}$ are matrices with pairwise orthonormal columns (i.e., $\bU^T\bU=\bI$ and $\bV^T\bV=\bI$)  that contain the left and right singular vectors of $\bA$ corresponding to the non-zero singular values; $\bSigma \in \Rs{\rho\times \rho}$ is a diagonal matrix with the non-zero singular values of $\bA$ in decreasing order on the diagonal.
\end{definition}
 If $\bA$ is a nonsingular matrix, we can compute its inverse using the SVD:
\[\bA^{-1} = (\bU\bSigma\bV^\top)^{-1} = \bV\bSigma^{-1}\bU^\top.\]
(If $\bA$ is nonsingular, then it is square and full rank, in which case the thin SVD is the same as the full SVD.)
The SVD is so important since, as is well-known, the best rank-$k$ approximation to any matrix can be computed via the SVD.
\begin{theorem}
	Let $\bA=\bU{\bSigma}\bV^\top \in \Rs{m\times n}$ be the thin SVD of $\bA$; let $k<\rank{\bA}=\rho$ be an integer; and let
$\bA_k = \sum_{i=1}^k\sigma_i\bu_i\bv_i^\top = \bU_k \bSigma_k \bV_k^T$. Then,
	\[\sigma_{k+1} = \min_{\bB\in\Rs{m\times n},\ \rm{rank}(\bB)=k}\TNorm{\bA-\bB} = \TNorm{\bA-\bA_k}\] and
	\[\sum_{j=k+1}^\rho\sigma_j^2 = \min_{\bB\in\Rs{m\times n},\ \rm{rank}(\bB)=k}\FNormS{\bA-\bB} = \FNormS{\bA-\bA_k}.\]
\end{theorem}
In words, the above theorem states that if we seek a rank $k$ approximation to a matrix $\bA$ that minimizes the two or the Frobenius norm of the ``error'' matrix, i.e., of the difference between $\bA$ and its approximation, then it suffices to keep the top $k$ singular values of $\bA$ and the corresponding left and right singular vectors.

We will often use the following notation: let $\bU_k \in \mathbb{R}^{m \times k}$ (respectively, $\bV_k \in \mathbb{R}^{n \times k}$) denote the matrix of the top $k$ left (respectively, right) singular vectors of $\bA$; and let $\bSigma_k \in \mathbb{R}^{k \times k}$ denote the diagonal matrix containing the top $k$ singular values of $\bA$. Similarly,
let $\bU_{k,\perp} \in \mathbb{R}^{m \times (\rho-k)}$ (respectively, $\bV_{k,\perp} \in \mathbb{R}^{n \times (\rho-k)}$) denote the matrix of the bottom $\rho-k$ nonzero left (respectively, right) singular vectors of $\bA$; and let $\bSigma_{k,\perp} \in \mathbb{R}^{(\rho-k) \times (\rho-k)}$ denote the diagonal matrix containing the bottom $\rho-k$ singular values of $\bA$. Then,
\begin{equation}\label{eqn:svdnotation}
\bA_k=\bU_k\bSigma_k\bV_k^T \quad \mbox{and} \quad \bA_{k,\perp}=\bA-\bA_k = \bU_{k,\perp}\bSigma_{k,\perp}\bV_{k,\perp}^T.
\end{equation}

\subsection{SVD and Fundamental Matrix Spaces.}
Any matrix $\bA \in \Rs{m\times n}$ defines four fundamental spaces:
\begin{description}
	\item [The Column Space of $\bA$] This space is spanned by the columns of $\bA$:
	\[\rm{range}(\bA) = \{\bvb: \bA\bx = \bvb,\quad \bx \in \Rs{n}\} \subset \Rs{m}.\]
	\item [The Null Space of $\bA$] This space is spanned by all vectors $\bx\in\Rs{n}$ such that $ \bA\bx = \bzero$:
	\[\rm{null}(\bA) = \{\bx: \bA\bx = \bzero\} \subset \Rs{n}.\]
	\item [The Row Space of $\bA$] This space is spanned by the rows of $\bA$:
	\[\rm{range}(\bA^\top) = \{\bd: \bA^\top\by = \bd,\quad \by \in \Rs{m}\} \subset \Rs{n}.\]
	\item [The Left Null Space of $\bA$] This space is spanned by all vectors $\by\in\Rs{m}$ such that $ \bA^\top\by = \bzero$:
	\[\rm{null}(\bA^\top) = \{\by: \bA^\top\by = \bzero\} \subset \Rs{m}.\]
\end{description}
The SVD reveals orthogonal bases for all these spaces. Given a matrix $\bA \in \Rs{m\times n}$, with $\rank{\bA} = \rho$, its SVD can be written as:
\[\bA = \begin{pmatrix}\bU_\rho & \bU_{\rho,\perp}\end{pmatrix}\begin{pmatrix}\bSigma_\rho & \bzero\\\bzero & \bzero\end{pmatrix}\begin{pmatrix}\bV_\rho^\top\\\bV_{\rho,\perp}^\top\end{pmatrix}.\]
It is easy to prove that:
\[ \range{\bA} = \range{\bU_\rho}, \]
\[ \rnull{\bA} = \range{\bV_{\rho,\perp}}, \]
\[ \range{\bA^\top} =\range{\bV_\rho}, \]
\[ \rnull{\bA^\top} = \range{\bU_{\rho,\perp}}.\]

\begin{theorem}[Basic Theorem of Linear Algebra.]
	The column space of $\bA$ is orthogonal to the null space of $\bA^\top$ and their union is $\Rs{m}$. The column space of $\bA^\top$ is orthogonal to the null space of $\bA$ and their union is $\Rs{n}$.
\end{theorem}

\subsection{Matrix Schatten norms.}\label{sxn:schatten}
The matrix Schatten norms are a special family of norms that are defined on the vector containing the singular values of a matrix.
Given a matrix $\bA\in\Rs{m\times n}$ with singular values $\sigma_1\geq\dots\geq\sigma_\rho >0$, we define the Schatten p-norm as:
\[\|\bA\|_p = \left(\sum_{i=1}^\rho\sigma_i^p\right)^{\frac{1}{p}}.\]
Common Schatten norms of a matrix $\bA \in \Rs{m\times n}$ are:
\begin{description}
	\item[Schatten one-norm] The nuclear norm, i.e., the sum of the singular values.
	\item [Schatten two-norm] The Frobenius norm, i.e., the square root of the sum of the squares of the singular values.
	\item [Schatten infinity-norm] The two norm, i.e., the largest singular value.
\end{description}
Schatten norms are orthogonally invariant, submultiplicative, and satisfy H\"{o}lder's inequality.

\subsection{The Moore-Penrose pseudoinverse.}\label{sxn:review:MP}
A generalized notion of the well-known matrix inverse is the Moore-Penrose pseudoinverse.
Formally, given a matrix $\bA\in\mathbb{R}^{m \times n}$, a matrix $\bA^\dagger$ is the Moore Penrose pseudoinverse of $\bA$ if it satisfies the following
properties:
\begin{enumerate}
\item $\bA\bA^\dagger\bA = \bA$.
\item $\bA^\dagger\bA\bA^\dagger = \bA^\dagger$.
\item $(\bA\bA^\dagger)^\top = \bA\bA^\dagger$.
\item $(\bA^\dagger\bA)^\top = \bA^\dagger\bA$.
\end{enumerate}
Given a matrix $\bA\in\Rs{m\times n}$ of rank $\rho$ and its thin SVD
	\[\bA = \sum_{i=1}^\rho \sigma_i\bu_i\bv_i^\top,\]
	its Moore-Penrose pseudoinverse $\bA^{\dagger}$ is
	\[{\bA}^\dagger = \sum_{i=1}^\rho\frac{1}{\sigma_i}\bv_i\bu_i^\top.\]
If a matrix $\bA\in\Rs{n\times n}$ has full rank, then $\bA^{\dagger} = \bA^{-1}$.
If a matrix $\bA\in\Rs{m\times n}$ has full column rank, then
$\bA^{\dagger}\bA = \bI_n$, and $\bA\bA^{\dagger}$ is a projection matrix onto the column span of $\bA$; while if it has full row rank, then
$\bA\bA^{\dagger} = \bI_m$, and $\bA^{\dagger}\bA$ is a projection matrix onto the row span of $\bA$.

A particularly important property regarding the pseudoinverse of the product of two matrices is the following: for matrices $\bY_1\in\mathbb{R}^{m\times p}$ and
$\bY_2\in\mathbb{R}^{p\times n}$,
satisfying $\rank{\bY_1}=\rank{\bY_2}$,~\cite[Theorem 2.2.3]{Bjo15} states that
\begin{align}\label{eqn:pinv}
\left(\bY_1\bY_2\right)^{\dagger} = \bY_2^{\dagger}\bY_1^{\dagger}.
\end{align}
(We emphasize that the condition on the ranks is crucial: while the inverse of the product of two matrices always equals the product of the inverse of those matrices, the analogous statement is not true in full generality for the Moore-Penrose pseudoinverse~\cite{Bjo15}.)
The fundamental spaces of the Moore-Penrose pseudoinverse are connected with those of the actual matrix.
Given a matrix $\bA$ and its Moore-Penrose pseudoinverse $\bA^\dagger$, the column space of $\bA^\dagger$ can be defined~as:
\[\rm{range}(\bA^\dagger) = \rm{range}(\bA^\top\bA) = \rm{range}(\bA^\top) , \]
and it is orthogonal to the null space of $\bA$. The null space of  $\bA^\dagger$ can be defined~as:
\[\rm{null}(\bA^\dagger) = \rm{null}(\bA\bA^\top) = \rm{null}(\bA^\top) , \]
and it is orthogonal to the column space of $\bA$.

\subsection{References.} We refer the interested reader to~\cite{Strang88,GVL96,TrefethenBau97,Bjo15} for additional background on linear algebra and matrix computations, as well as to~\cite{Stewart90,Bhatia97} for additional background on matrix perturbation theory. 

\section{Discrete Probability}\label{sxn:dp}

In this section, we present a brief overview of discrete probability. More advanced results (in particular, Bernstein-type inequalities for real-valued and matrix-valued random variables) will be introduced in the appropriate context later in the chapter. It is worth noting that most of RandNLA builds upon simple, fundamental principles of discrete (instead of continuous) probability.

\subsection{Random experiments: basics.}
A random experiment is any procedure that can be infinitely repeated and has a well-defined set of possible outcomes. Typical examples are the roll of a dice or the toss of a coin. The sample space $\Omega$ of a random experiment is the set of all possible outcomes of the random experiment. If the random experiment only has two possible outcomes (e.g., success and failure) then it is often called a Bernoulli trial. In discrete probability, the sample space $\Omega$ is finite. (We will \textit{not} cover countably or uncountably infinite sample spaces in this chapter.)

An event is any subset of the sample space $\Omega$. Clearly, the set of all possible events is the powerset (the set of all possible subsets) of $\Omega$, often denoted as $2^{\Omega}$. As an example, consider the following random experiment: toss a coin three times. Then, the sample space $\Omega$ is
\[\Omega = \{HHH, HHT, HTH, HTT, THH, THT, TTH, TTT\}\]
and an event ${\mathcal E}$ could be described in words as ``the output of the random experiment was either all heads or all tails''. Then,
${\mathcal E} = \{HHH,TTT\}.$
The \textit{probability measure} or \textit{probability function} maps the (finite) sample space $\Omega$ to the interval $[0,1]$. Formally, let the function $\Probab{\omega}$ for all $\omega \in \Omega$ be a function whose domain is $\Omega$ and whose range is the interval $[0,1]$. This function has the so-called normalization property, namely
\[\sum_{\omega \in \Omega} \Probab{\omega} = 1.\]
If $\mathcal E$ is an event, then
\begin{equation}\label{eqn:setprop1}
\Probab{\mathcal E} = \sum_{\omega \in {\mathcal E}} \Probab{\omega},
\end{equation}
namely the probability of an event is the sum of the probabilities of its elements. It follows that the probability of the empty event (the event ${\mathcal E}$ that corresponds to the empty set) is equal to zero, whereas the probability of the event $\Omega$ (clearly $\Omega$ itself is an event) is equal to one. Finally, the uniform probability function is defined as
$\Probab{\omega} = 1/\abs{\Omega}$,
for all $\omega \in \Omega$.

\subsection{Properties of events.}
Recall that events are sets and thus set operations (union, intersection, complementation) are applicable. Assuming finite sample spaces and using Eqn.~(\ref{eqn:setprop1}), it is easy to prove the following property for the union of two events ${\mathcal E}_1$ and ${\mathcal E}_2$:
\[\Probab{{\mathcal E}_1 \cup {\mathcal E}_2} = \Probab{{\mathcal E}_1} + \Probab{{\mathcal E}_2} -  \Probab{{\mathcal E}_1\cap {\mathcal E}_2}.\]
This property follows from the well-known inclusion-exclusion principle for set union and can be generalized to more than two sets and thus to more than two events. Similarly, one can prove that
$\Probab{\bar{\mathcal E}} = 1-\Probab{{\mathcal E}}.$
In the above, $\bar{\mathcal E}$ denotes the complement of the event $\mathcal E$. Finally, it is trivial to see that if ${\mathcal E}_1$ is a subset of ${\mathcal E}_2$ then
$\Probab{{\mathcal E}_1} \leq \Probab{{\mathcal E}_2}.$

\subsection{The union bound.}
The union bound is a fundamental result in discrete probability and can be used to bound the probability of a union of events without any special assumptions on the relationships between the events. Indeed, let ${\mathcal E}_i$ for all $i=1,\ldots,n$ be events defined over a finite sample space $\Omega$. Then, the union bound states that
\[\Probab{\bigcup_{i=1}^n{\mathcal E}_i} \leq \sum_{i=1}^n \Probab{{\mathcal E}_i}.\]
The proof of the union bound is quite simple and can be done by induction, using the inclusion-exclusion principle for two sets that was discussed in the previous section.

\subsection{Disjoint events and independent events.}
Two events ${\mathcal E}_1$ and ${\mathcal E}_2$ are called \textit{disjoint} or \textit{mutually exclusive} if their intersection is the empty set, i.e., if 
\[  {\mathcal E}_1 \cap {\mathcal E}_2 = \emptyset  . \]
This can be generalized to any number of events by necessitating that the events are all pairwise disjoint.
Two events ${\mathcal E}_1$ and ${\mathcal E}_2$ are called \textit{independent} if the occurrence of one does not affect the probability of the other. Formally, they must satisfy
\[\Probab{{\mathcal E}_1 \cap {\mathcal E}_2} = \Probab{{\mathcal E}_1}\cdot \Probab{{\mathcal E}_2}.\]
Again, this can be generalized to more than two events by necessitating that the events are all pairwise independent.

\subsection{Conditional probability.}
For any two events ${\mathcal E}_1$ and ${\mathcal E}_2$, the conditional probability $\Probab{{\mathcal E}_1|{\mathcal E}_2}$ is the probability that ${\mathcal E}_1$ occurs given that ${\mathcal E}_2$ occurs. Formally,
\[\Probab{{\mathcal E}_1|{\mathcal E}_2} = \frac{\Probab{{\mathcal E}_1\cap{\mathcal E}_2}}{\Probab{{\mathcal E}_2}}.\]
Obviously, the probability of ${\mathcal E}_2$ in the denominator must be non-zero for this to be well-defined.
The well-known Bayes rule states that for any two events ${\mathcal E}_1$ and ${\mathcal E}_2$ such that $\Probab{{\mathcal E}_1}>0$ and $\Probab{{\mathcal E}_2}>0$,
\[\Probab{{\mathcal E}_2|{\mathcal E}_1} = \frac{\Probab{{\mathcal E}_1|{\mathcal E}_2}\Probab{{\mathcal E}_2}}{\Probab{{\mathcal E}_1}}.\]
Using the Bayes rule and the fact that the sample space $\Omega$ can be partitioned as $\Omega = {\mathcal E}_2 \cup \overline{{\mathcal E}_2}$, it follows that
\[\Probab{{\mathcal E}_1} = \Probab{{\mathcal E}_1|{\mathcal E}_2}\Probab{{\mathcal E}_2} + \Probab{{\mathcal E}_1|\overline{{\mathcal E}_2}}\Probab{\overline{{\mathcal E}_2}}.\]
We note that the probabilities of both events ${\mathcal E}_1$ and ${\mathcal E}_2$ must be in the open interval $(0,1)$.
We can now revisit the notion of independent events. Indeed, for any two events ${\mathcal E}_1$ and ${\mathcal E}_2$ such that $\Probab{{\mathcal E}_1}>0$ and $\Probab{{\mathcal E}_2}>0$ the following statements are equivalent:
\begin{enumerate}
\item $\Probab{{\mathcal E}_1|{\mathcal E}_2} = \Probab{{\mathcal E}_1}$,
\item $\Probab{{\mathcal E}_2|{\mathcal E}_1} = \Probab{{\mathcal E}_2}$, and
\item $\Probab{{\mathcal E}_1 \cap {\mathcal E}_2} = \Probab{{\mathcal E}_1}\Probab{{\mathcal E}_2}$.
\end{enumerate}
Recall that the last statement was the definition of independence in the previous section.

\subsection{Random variables.}
Random variables are \textit{functions} mapping the sample space $\Omega$ to the real numbers $\mathbb{R}$. Note that even though they are called variables, in reality they are functions. Let $\Omega$ be the sample space of a random experiment. A formal definition for the random variable $X$ would be as follows: let $\alpha\in\mathbb{R}$ be a real number (not necessarily positive) and note that the function
\[X^{-1}\left(\alpha\right) = \left\{\omega \in \Omega\ :\ X\left(\omega\right)=\alpha\right\}\]
returns a subset of $\Omega$ and thus is an event. Therefore, the function $X^{-1}\left(\alpha\right)$ has a probability. We will abuse notation and write:
\[\Probab{X = \alpha}\]
instead of the more proper notation $\Probab{X^{-1}\left(\alpha\right)}$. This function of $\alpha$ is of great interest and it is easy to generalize as follows:
\[\Probab{X \leq \alpha} = \Probab{X^{-1}\left(\alpha^{\prime}\right): \alpha^{\prime} \in (-\infty,\alpha]} = \Probab{\omega \in \Omega\ :\ X\left(\omega\right)\leq \alpha}.\]
%

\subsection{Probability mass function and cumulative distribution function.}
Two common functions associated with random variables are the probability mass function (PMF) and the cumulative distribution function (CDF). The first measures the probability that a random variable takes a particular value $\alpha\in\mathbb{R}$, and the second measures the probability that a random variable takes any value below $\alpha\in\mathbb{R}$.

\begin{definition}[Probability Mass Function (PMF)]
Given a random variable $X$ and a real number $\alpha$, the \emph{probability mass function (PMF)} is the function
$f(\alpha) = \Probab{X=\alpha}.$
\end{definition}

\begin{definition}[Cumulative Distribution Function (CDF)]
Given a random variable $X$ and a real number $\alpha$, the \emph{cumulative distribution function (CDF)} is the function
$F(\alpha) = \Probab{X\leq\alpha}.$
\end{definition}
It is obvious from the above definitions that
$F(\alpha) =  \sum_{x\leq \alpha} f(x).$

\subsection{Independent random variables.}
Following the notion of independence for events, we can now define the notion of independence for random variables. Indeed, two random variables $X$ and $Y$ are independent if for all reals $a$ and $b$,
\[\Probab{X=a\ \mbox{and}\ Y=b} = \Probab{X=a}\cdot\Probab{Y=b}.\]

\subsection{Expectation of a random variable.}
Given a random variable $X$, its expectation $\Expect{X}$ is defined as
\[\Expect{X} = \sum_{x \in X(\Omega)} x \cdot \Probab{X=x}.\]
In the above, $X(\Omega)$ is the image of the random variable $X$ over the sample space $\Omega$; recall that $X$ is a function. 
That is, the sum is over the range of the random variable $X$.
Alternatively, $\Expect{X}$ can be expressed in terms of a sum over the domain of $X$, i.e., over $\Omega$.
For finite sample spaces $\Omega$, such as those that arise in discrete probability, we get
\[\Expect{X} = \sum_{\omega \in \Omega} X(\omega) \Probab{\omega}.\]
We now discuss fundamental properties of the expectation. The most important property is linearity of expectation: for any random variables $X$ and $Y$ and real number $\lambda$,
\begin{align*}
\Expect{X+Y} &= \Expect{X} +\Expect{Y},\ \mbox{and}\\
\Expect{\lambda X} &= \lambda\Expect{X}.
\end{align*}
The first property generalizes to any finite sum of random variables and does not need any assumptions on the random variables involved in the summation.
If two random variables $X$ and $Y$ are independent then we can manipulate the expectation of their product as follows:
\[\Expect{XY} = \Expect{X}\cdot\Expect{Y}.\]

\subsection{Variance of a random variable.}
Given a random variable $X$, its variance $\Var{X}$ is defined as
\[\Var{X} = \Expect{X-\Expect{X}}^2.\]
In words, the variance measures the average of the (square) of the difference $X - \Expect{X}$.
The standard deviation is the square root of the variance and is often denoted by $\sigma$. It is easy to prove that
\[\Var{X} = \Expect{X^2} - \Expect{X}^2.\]
This obviously implies
\[\Var{X} \leq \Expect{X^2},\]
which is often all we need in order to get an upper bound for the variance. Unlike the expectation, the variance does not have a linearity property, unless the random variables involved are independent. Indeed, if the random variables $X$ and $Y$ are independent, then
\begin{align*}
\Var{X+Y} = \Var{X} +\Var{Y}.
\end{align*}
The above property generalizes to sums of more than two random variables, assuming that all involved random variables are pairwise independent. Also, for any real $\lambda$,
\begin{align*}
\Var{\lambda X} = \lambda^2\Var{X}.
\end{align*}

\subsection{Markov's inequality.}\label{sxn:review:Markov}
Let $X$ be a non-negative random variable; for any $\alpha > 0$,
\[\Probab{X\geq \alpha}\leq \frac{\Expect{X}}{\alpha}.\]
This is a very simple inequality to apply and only needs an upper bound for the expectation of $X$. An equivalent formulation is the following: let $X$ be a non-negative random variable; for any $k > 1$,
\[\Probab{X\geq k\cdot \Expect{X}}\leq \frac{1}{k},\]
or, equivalently,
\[\Probab{X\leq k\cdot \Expect{X}}\geq 1-\frac{1}{k}.\]
In words, the probability that a random variable exceeds $k$ times its expectation is at most $1/k$.
In order to prove Markov's inequality, we will show, \[\Probab{X \geq t} \leq \frac{\Expect{X}}{t}\] assuming \[k = \frac{t}{\Expect{X}},\] for any $t > 0$. In order to prove the above inequality, we define the following function \[f(X) =
\begin{cases}
1,		& \text{if }X \geq t\\
0, 		& \text{otherwise}
\end{cases}\]
with expectation:
\begin{align*}
\Expect{f(X)} = 1 \cdot \Probab{X \geq t} + 0 \cdot \Probab{X < t}=\Probab{X \geq t}.
\end{align*}
Clearly, from the function definition,  $f(X) \leq \frac{X}{t}$. Taking expectation on both sides:
\[\Expect{f(X)} \leq \Expect{\frac{X}{t}} = \frac{\Expect{X}}{t}.\]
Thus, \[\Probab{X \geq t} \leq \frac{\Expect{X}}{t}.\]
Hence, we conclude the proof of Markov's inequality.

\subsection{The Coupon Collector Problem.}\label{sxn:couponcollector}
Suppose there are $m$ types of coupons and we seek to collect them in independent trials, where in each trial the probability of obtaining any one coupon is $1/m$ (uniform). Let $X$ denote the number of trials that we need in order to collect at least one coupon of each type. Then, one can prove that~\cite[Section 3.6]{MotwaniRaghavan95}:
\begin{align*}
\Expect{X} &= m \ln m + \Theta\left(m\right),\ \mbox{and}\\
\Var{X} &= \frac{\pi^2}{6} m^2 + \Theta\left(m \ln m\right).
\end{align*}
%
%
The occurrence of the additional $\ln m$ factor in the expectation is common in sampling-based approaches that attempt to recover $m$ different types of objects using sampling in independent trials. Such factors will appear in many RandNLA sampling-based algorithms.

\subsection{References.} There are numerous texts covering discrete probability; most of the material in this chapter was adapted from~\cite{MotwaniRaghavan95}.

 \section{Randomized Matrix Multiplication}\label{chapter:MM}
Our first randomized algorithm for a numerical linear algebra problem is a simple, sampling-based approach to approximate the product of two matrices $\bA \in \mathbb{R}^{m \times n}$ and $\bB \in \mathbb{R}^{n \times p}$. This randomized matrix multiplication algorithm is at the heart of all of the RandNLA algorithms that we will discuss in this chapter, and indeed all of RandNLA more generally.
 It is of interest both pedagogically and in and of itself, and it is also used in an essential way in the analysis of the least squares approximation and low-rank approximation algorithms discussed below.

We start by noting that the product $\bA\bB$ may be written as the sum of $n$ rank one matrices:
 \begin{equation}
 \bA\bB = \sum_{t=1}^{n} \underbrace{\bA_{*t} \bB_{t*}}_{\in \mathbb{R}^{m \times n}},
 \label{AB_sum_rank_one}
 \end{equation}
 where each of the summands is the \emph{outer product} of a column of $\bA$ and the corresponding row of $\bB$. Recall that the standard definition of matrix multiplication states that the $(i,j)$-th entry of the matrix product $\bA\bB$ is equal to the \emph{inner product} of the $i$-th row of $\bA$ and the $j$-th column of $\bB$, namely
\[(\bA\bB)_{ij} = \bA_{i*} \bB_{*j} \in \mathbb{R}.\]
 It is easy to see that the two definitions are equivalent. However, when matrix multiplication is formulated as in Eqn.~(\ref{AB_sum_rank_one}), a simple randomized algorithm to approximate the product $\bA\bB$ suggests itself: in independent identically distributed (i.i.d.) trials, randomly sample (and appropriately rescale) a few rank-one matrices from the $n$ terms in the summation of Eqn.~(\ref{AB_sum_rank_one}); and then output the sum of the (rescaled) terms as an estimator for $\bA \bB$.
 %
 \input{fig_matmul1_alg}

Consider the \textsc{RandMatrixMultiply} algorithm (Algorithm~\ref{fig:BasicMatrixMultiplicationAlgorithm}), which makes this simple idea precise.
 When this algorithm is given as input two matrices $\bA$ and $\bB$, a probability
 distribution $\left\{ p_k \right\}_{k=1}^{n}$, and a number $c$ of column-row
 pairs to choose, it returns as output an estimator for the product $\bA\bB$ of the form
 \[\sum_{t=1}^{c} \frac{1}{cp_{i_t}} \bA_{*i_t} \bB_{i_t *}.\]
 Equivalently, the above estimator can be thought of as the product of the two matrices
$\bC$ and $\bR$ formed by the \textsc{RandMatrixMultiply} algorithm, where $\bC$ consists of $c$ (rescaled) columns of $\bA$ and $\bR$ consists of the corresponding (rescaled) rows of $\bB$. Observe that
 \[
 \bC\bR = \sum_{t=1}^{c} \bC_{*t} \bR_{t*}
 = \sum_{t=1}^{c} \left(\sqrt{\frac{1}{cp_{i_t}}} \bA_{*i_t}\right)\left( \sqrt{\frac{1}{cp_{i_t}}}\bB_{i_t *}\right)
 = \frac{1}{c} \sum_{t=1}^{c} \frac{1}{p_{i_t}} \bA_{*i_t} \bB_{i_t *}.
 \]
Therefore, the procedure used for sampling and scaling column-row pairs in the \textsc{RandMatrixMultiply} algorithm corresponds to sampling and rescaling terms in Eqn.~(\ref{AB_sum_rank_one}).

\begin{Remark}
 The analysis of RandNLA algorithms has benefited enormously from formulating algorithms using the so-called \emph{sampling-and-rescaling matrix formalism}. Let's define the sampling-and-rescaling matrix
 $\bS \in \Rs{n \times c}$ to be a matrix with
 $\bS_{i_t t} =  1/\sqrt{cp_{i_t}}$ if the $i_t$-th column of $\bA$
 is chosen in the $t$-th trial (all other entries of $\bS$ are set to zero).
Then
 \[
 \bC = \bA \bS \mbox{ and } \bR = \bS^{T} \bB,
 \]
 so that $ \bC\bR = \bA \bS \bS^{T} \bB \approx \bA\bB $.
 Obviously, the matrix $\bS$ is very sparse, having a single non-zero entry per column, for a total of $c$ non-zero entries, and so it is not explicitly constructed and stored by the algorithm.
\end{Remark}

\begin{Remark}
The choice of the sampling probabilities $\left\{ p_k \right\}_{k=1}^{n}$ in the \textsc{RandMatrixMultiply}
algorithm is very important. As we will prove in Lemma~\ref{lem:multexpvar}, the estimator returned by the \textsc{RandMatrixMultiply} algorithm is (in an element-wise sense) unbiased, regardless of our choice of the sampling probabilities. However, a natural notion of the \textit{variance} of our estimator (see Theorem~\ref{lem:basicmult} for a precise definition) is minimized when the sampling probabilities are set to
 \[p_k = \frac{\FNorm{\bA_{*k} \bB_{k*}}}{\sum_{k'=1}^n\FNorm{\bA_{*k'} \bB_{k'*}}}  = \frac{\TNorm{\bA_{*k}} \TNorm{\bB_{k*}}}{\sum_{k'=1}^n\TNorm{\bA_{*k'}} \TNorm{\bB_{k'*}}}.\]
 In words, the best choice when sampling rank-one matrices from the summation of Eqn.~(\ref{AB_sum_rank_one}) is to select rank-one matrices that have larger Frobenius norms with higher probabilities. This is equivalent to selecting column-row pairs that have larger (products of) Euclidean norms with higher probability.
\end{Remark}

\begin{Remark} This approach for approximating matrix multiplication has several advantages. First, it is conceptually simple. Second, since the heart of the algorithm involves matrix multiplication of smaller matrices, it can use any algorithms that exist in the literature for performing the desired matrix multiplication. Third, this approach does not tamper with the sparsity of the input matrices. Finally, the algorithm can be easily implemented in one pass over the input matrices $\bA$ and $\bB$, given the sampling probabilities $\left\{ p_k \right\}_{k=1}^{n}$. See~\cite[Section 4.2]{dkm_matrix1} for a detailed discussion regarding the implementation of the~\textsc{RandMatrixMultiply} algorithm in the pass-efficient and streaming models of computation.
\end{Remark}

\subsection{Analysis of the {\small R{\scriptsize AND}M{\scriptsize ATRIX}M{\scriptsize ULTIPLY}} algorithm.}
 \label{sxn:matmult:analysis}
 This section provides upper bounds for the error matrix $\FNormS{\bA\bB-\bC\bR}$, where $\bC$ and $\bR$ are the outputs of the \textsc{RandMatrixMultiply} algorithm.

 Our first lemma proves that the expectation of the $(i,j)$-th element of the
 estimator $\bC\bR$ is equal to the $(i,j)$-th element of the exact product $\bA \bB$, regardless of the choice of the sampling probabilities. It also
 bounds the variance of the $(i,j)$-th element of the estimator, which does depend on our choice of the sampling probabilities.
 \begin{lemma}
 	\label{lem:multexpvar}
 	Let $\bC$ and $\bR$ be constructed as described in the \textsc{RandMatrixMultiply} algorithm.
 	Then,
 	\[
 	\Expect{(\bC\bR)_{ij}}=(\bA\bB)_{ij}
 	\]
 	and
 	\[
 	\Var{(\bC\bR)_{ij}} \leq \frac{1}{c}\sum_{k=1}^n \frac{\bA_{ik}^2 \bB_{kj}^2}{p_k}.
 	\]
 \end{lemma}
 \begin{proof}
 	Fix some pair $i,j$. For $t=1,\dots,c$, define
 	$
 	X_t = \left( \frac{ \bA_{*i_t}\bB_{i_t *} }{ cp_{i_t} } \right)_{ij}
 	=        \frac{ \bA_{ii_t}\bB_{i_tj}  }{ cp_{i_t} }
 	$.
 	Thus, for any $t$,
 	\[
 	\Expect{X_t}   = \sum_{k=1}^n p_k \frac{ \bA_{ik}\bB_{kj} }{ cp_k }
 	= \frac{1}{c} \sum_{k=1}^{n}\bA_{ik}\bB_{kj}
 	= \frac{1}{c}(\bA\bB)_{ij}.
 	\]
 	Since we have $(\bC\bR)_{ij} = \sum_{t = 1}^{c}X_t$, it follows that
 	\[
 	\Expect{(\bC\bR)_{ij}} = \Expect{\sum_{t = 1}^{c}X_t} = \sum_{t = 1}^{c}\Expect{X_t} = (\bA\bB)_{ij}.
 	\]
 	Hence, $\bC\bR$ is an unbiased estimator of $\bA\bB$, regardless of the choice of the sampling probabilities. Using the fact that $(\bC\bR)_{ij}$ is the sum of $c$ independent random variables, we get
 	\[
 	\Var{(\bC\bR)_{ij}} = \Var{\sum_{t=1}^c X_t} = \sum_{t=1}^c\Var{X_t}.
 	\]
 	Using
 	$
 	\Var{X_t} \leq \Expect{X_t^2}= \sum_{k=1}^n \frac{ \bA_{ik}^2 \bB_{kj}^2 }{ c^2 p_k },
 	$ we get
 	\[
 	\Var{(\bC\bR)_{ij}} = \sum_{t=1}^c\Var{X_t} \leq
 	c \sum_{k=1}^n \frac{ \bA_{ik}^2 \bB_{kj}^2 }{ c^2 p_k } =
 	\frac{1}{c} \sum_{k=1}^n \frac{\bA_{ik}^2 \bB_{kj}^2}{p_k},
 	\]
which concludes the proof of the lemma.
 \end{proof}

Our next result bounds the expectation of the Frobenius norm of the error matrix $\bA\bB-\bC\bR$. Notice that this error metric depends on our choice of the sampling probabilities $\left\{ p_k \right\}_{k=1}^{n}$.

 \begin{theorem}
 	\label{lem:basicmult}
 	Construct $\bC$ and $\bR$ using the \textsc{RandMatrixMultiply} algorithm
 	and let $\bC\bR$ be an approximation to $\bA\bB$.
 	Then,
 	\begin{equation}
 	\Expect{\FNormS{\bA\bB-\bC\bR}}
 	\leq \sum_{k=1}^n \frac{\TNormS{\bA_{*k}}\TNormS{\bB_{k*}}}{c p_k}.
 	\end{equation}
 	Furthermore, if
 	\begin{equation}
 	p_k = \frac{ \TNorm{\bA_{*k}} \TNorm{\bB_{k*}} }{ \sum_{k^\prime=1}^n \TNorm{\bA_{*k^\prime}} \TNorm{\bB_{k^\prime *}} }   ,
 	\label{optimal_probs}
 	\end{equation}
 	for all $k = 1,\ldots,n$,
 	then
 	\begin{equation}
 	\Expect{\FNormS{\bA\bB-\bC\bR}} \leq \frac{1}{c}\left( \sum_{k=1}^n \TNorm{\bA_{*k}}\TNorm{\bB_{k*}} \right)^2.
 	\end{equation}
 	This choice for $\left\{ p_k \right\}_{k=1}^{n}$ minimizes $\Expect{\FNormS{\bA\bB-\bC\bR}}$, among possible choices for the sampling probabilities.
\end{theorem}

 \begin{proof}
 	First of all, since $\bC\bR$ is an unbiased estimator of $\bA\bB$, $\Expect{(\bA\bB-\bC\bR)_{ij}} = 0$. Thus,
 	\[
 	\Expect{\FNormS{\bA\bB-\bC\bR}} = \sum_{i=1}^m \sum_{j=1}^p \Expect{\left(\bA\bB-\bC\bR\right)_{ij}^2}
 	= \sum_{i=1}^m \sum_{j=1}^p \Var{(\bC\bR)_{ij}}     .
 	\]
 	Using Lemma~\ref{lem:multexpvar}, we get
 	\begin{align*}
 		\Expect{\FNormS{\bA\bB-\bC\bR}}
 		&\leq \frac{1}{c}\sum_{k=1}^n \frac{1}{p_k} \left(\sum_i \bA_{ik}^2\right)\left(\sum_j \bB_{kj}^2\right) \\
 		&= \frac{1}{c}\sum_{k=1}^n \frac{1}{p_k} \TNormS{\bA_{*k}} \TNormS{\bB_{k*}}.
 	\end{align*}
Let $p_k$ be as in Eqn.~(\ref{optimal_probs}); then
 	\begin{align*}
 		\Expect{\FNormS{\bA\bB-\bC\bR}} \leq \frac{1}{c}\left(\sum_{k=1}^n \TNorm{\bA_{*k}} \TNorm{\bB_{k*}}\right)^2.
 	\end{align*}
 	Finally, to prove that the aforementioned choice for the $\left\{ p_k \right\}_{k=1}^{n}$ minimizes the quantity
 	$\Expect{\FNormS{\bA\bB-\bC\bR}}$, define the function
 	\[
 	f(p_1,\dots p_n) = \sum_{k=1}^n \frac{1}{p_k}\TNormS{\bA_{*k}}\TNormS{\bB_{k*}},
 	\]
 	which characterizes the dependence of $\Expect{\FNormS{\bA\bB-\bC\bR}}$ on the $p_k$'s. In order to minimize $f$ subject to $\sum_{k=1}^n p_k =1$, we can introduce the Lagrange multiplier $\lambda$ and define the function
 	\[
 	g(p_1,\dots p_n) = f(p_1,\dots p_n) + \lambda\left(\sum_{k=1}^n p_k-1\right)  .
 	\]
 	We then have the minimum at
 	\[
 	0 = \frac{\partial g}{\partial p_k}
 	= \frac{-1}{p_k^2}\TNormS{\bA_{*k}}\TNormS{\bB_{k*}} + \lambda   .
 	\]
 	Thus,
 	\[
 	p_k = \frac{\TNorm{\bA_{*k}}\TNorm{\bB_{k*}}}{\sqrt{\lambda}}
 	= \frac{\TNorm{\bA_{*k}}\TNorm{\bB_{k*}}}
 	{\sum_{k^\prime=1}^n\TNorm{\bA_{*k^\prime}}\TNorm{\bB_{k^\prime *}}} ,
 	\]
 	where the second equality comes from solving for $\sqrt{\lambda}$ in
 	$\sum_{k=1}^{n} p_k = 1$.
 	These probabilities are minimizers of $f$ because $\frac{\partial^2 g}{{\partial p_k}^2} > 0$ for all
 	$k$.
 \end{proof}

We conclude this section by pointing out that we can apply Markov's inequality on the expectation bound of Theorem~\ref{lem:basicmult} in order to get bounds for the Frobenius norm of the error matrix $\bA\bB-\bC\bR$ that hold with constant probability. We refer the reader to~\cite[Section 4.4]{dkm_matrix1} for a tighter analysis, arguing for a better (in the sense of better dependence on the failure probability than provided by Markov's inequality) concentration of the Frobenius norm of the error matrix around its mean
using a martingale argument.

 \subsection{Analysis of the algorithm for nearly optimal probabilities.}
 \label{sxn:matmult:analysis:probs}

 We now discuss three different choices for the sampling probabilities that are easy to analyze and will be useful in this chapter. We summarize these results in the following list; all three bounds can be easily proven following the proof of Theorem~\ref{lem:basicmult}.
 \begin{description}
 \item[Nearly optimal probabilities, depending on both $\bA$ and $\bB$] Let the $\left\{ p_k \right\}_{k=1}^{n}$ satisfy
 \begin{equation}\label{eqn:appopt1}
 \sum_{k=1}^n p_k=1 \quad \mbox{and} \quad
 p_k \ge \frac{ \beta\TNorm{\bA_{*k}} \TNorm{\bB_{k*}} }{ \sum_{k^\prime=1}^n \TNorm{\bA_{*k^\prime}} \TNorm{\bB_{k^\prime *}} },
 \end{equation}
 for some positive constant $\beta \le 1$. Then,
 \begin{equation}\label{eqn:appopt1result}
 \Expect{\FNormS{\bA\bB-\bC\bR}} \leq \frac{1}{\beta c}\left( \sum_{k=1}^n \TNorm{\bA_{*k}}\TNorm{\bB_{k*}} \right)^2.
 \end{equation}
 \item[Nearly optimal probabilities, depending only on $\bA$] Let the $\left\{ p_k \right\}_{k=1}^{n}$ satisfy
 \begin{equation}\label{eqn:appopt2}
 \sum_{k=1}^n p_k=1 \quad \mbox{and} \quad
 p_k \ge \frac{ \beta\TNormS{\bA_{*k}}}{\FNormS{\bA}},
 \end{equation}
 for some positive constant $\beta \le 1$. Then,
 \begin{equation}\label{eqn:appopt2result}
 	\Expect{\FNormS{\bA\bB-\bC\bR}} \leq \frac{1}{\beta c}\FNormS{\bA}\FNormS{\bB}.
 \end{equation}
 \item[Nearly optimal probabilities, depending only on $\bB$] Let  the $\left\{ p_k \right\}_{k=1}^{n}$  satisfy
 \begin{equation}\label{eqn:appopt3}
 \sum_{k=1}^n p_k=1 \quad \mbox{and} \quad
 p_k \ge \frac{ \beta\TNormS{\bB_{k*}}}{\FNormS{\bB}},
 \end{equation}
 for some positive constant $\beta \le 1$. Then,
 \begin{equation}\label{eqn:appopt3result}
 	\Expect{\FNormS{\bA\bB-\bC\bR}} \leq \frac{1}{\beta c}\FNormS{\bA}\FNormS{\bB}.
 \end{equation}
\end{description}
We note that, from the Cauchy-Schwartz inequality,
\[\left( \sum_{k=1}^n \TNorm{\bA_{*k}}\TNorm{\bB_{k*}} \right)^2 \leq \FNormS{\bA}\FNormS{\bB},\]
and thus the bound of Eqn.~(\ref{eqn:appopt1result}) is generally better than the bounds of Eqns.~(\ref{eqn:appopt2result}) and~(\ref{eqn:appopt3result}). See~\cite[Section 4.3, Table 1]{dkm_matrix1} for other sampling probabilities and respective error bounds that might be of interest.

\subsection{Bounding the two norm.}\label{sxn:chapter1:spectral}
In both applications of the \textsc{RandMatrixMultiply} algorithm that we will discuss in this chapter (see least-squares approximation and low-rank matrix approximation in Sections~\ref{sxn:main:regression} and~\ref{sxn:main:lowrank}, respectively), we will be particularly interested in approximating the product $\bU^T \bU$, where $\bU$ is a tall-and-thin matrix, by sampling (and rescaling) a few rows of $\bU$. 
(The matrix $\bU$ will be a matrix spanning the column space or the ``important'' part of the column space of some other matrix of interest.)
It turns out that, without loss of generality, 
we can focus on the special case where $\bU \in \mathbb{R}^{n \times d}$ ($n \gg d$) is a matrix with orthonormal columns (i.e., $\bU^T\bU = \bI_d$). 
Then, if we let $\bR \in \mathbb{R}^{c \times d}$ be a sample of $c$ (rescaled) rows of $\bU$ constructed using the \textsc{RandMatrixMultiply} algorithm, and note that the corresponding $c$ (rescaled) columns of $\bU^T$ form the matrix $\bR^T$, then Theorem~\ref{lem:basicmult} implies that
 \begin{equation}\label{eqn:UUT1}
   \Expect{\FNormS{\bU^T\bU-\bR^T\bR}}=\Expect{\FNormS{\bI_d-\bR^T\bR}} \leq \frac{d^2}{\beta c}.
 \end{equation}
 In the above, we used the fact that $\FNormS{\bU}=d$. For the above bound to hold, it suffices to use sampling probabilities $p_k$ ($k=1,\ldots, n$) that satisfy
 \begin{equation}\label{eqn:appopt4}
 \sum_{k=1}^n p_k=1 \quad \mbox{and} \quad
 p_k \ge \frac{ \beta\TNormS{\bU_{k*}}}{d}.
 \end{equation}
 (The quantities $\TNormS{\bU_{k*}}$ are known as leverage scores~\cite{Mah-mat-rev_BOOK}; and the probabilities given by Eqn.~(\ref{eqn:appopt4}) are nearly-optimal, in the sense of Eqn.~(\ref{eqn:appopt1}), i.e., in the sense that they approximate the optimal probabilities for approximating the matrix product shown in Eqn~(\ref{eqn:UUT1}), up to a $\beta$ factor.)
 Applying Markov's inequality to the bound of Eqn.~(\ref{eqn:UUT1}) and setting
 \begin{equation}\label{eqn:cval1}
  c = \frac{10d^2}{\beta\epsilon^2},
 \end{equation}
 we get that, with probability at least 9/10,
 \begin{equation}\label{eqn:UUT2}
   \FNorm{\bU^T\bU-\bR^T\bR}=\FNorm{\bI_d-\bR^T\bR}\leq \epsilon.
 \end{equation}
 Clearly, the above equation also implies a two-norm bound.
 Indeed, with probability at least 9/10, 
 $$
 \TNorm{\bU^T\bU-\bR^T\bR} = \TNorm{\bI_d-\bR^T\bR} \leq \epsilon
 $$ 
 by setting $c$ to the value of Eqn.~(\ref{eqn:cval1}). 

In the remainder of this section, we will state and prove a theorem that also guarantees $\TNorm{\bU^T\bU-\bR^T\bR}\leq\epsilon$, while setting $c$ to a value that is \textit{smaller} than the one in Eqn.~(\ref{eqn:cval1}).
For related concentration techniques, see the chapter by Vershynin in this volume~\cite{pcmi-chapter-vershynin}.

 \begin{theorem}\label{thm:theorem7correct}
 	Let $\bU \in \Rs{n \times d}$ ($n \gg d$) satisfy $\bU^T \bU = \bI_d$. Construct $\bR$ using the \textsc{RandMatrixMultiply} algorithm and let the sampling probabilities $\left\{ p_k \right\}_{k=1}^{n}$ satisfy the conditions of Eqn.~(\ref{eqn:appopt4}), for all $k=1,\ldots, n$ and some constant $\beta \in (0,1]$. Let $\epsilon \in (0,1)$ be an accuracy parameter and let
 	\begin{equation}\label{eqn:CboundAppendix}
 	c \geq \frac{96d}{\beta \epsilon^2}\ln \left(\frac{96d}{\beta \epsilon^2 \sqrt{\delta}}\right).
 	\end{equation}
 	Then, with probability at least $1-\delta$,
 	\[\TNorm{\bU^T\bU-\bR^T\bR}=\TNorm{\bI_d-\bR^T\bR}\leq \epsilon.\]
 \end{theorem}
Prior to proving the above theorem, we state a matrix-Bernstein inequality that is due to Oliveira~\cite[Lemma 1]{Oli10}.
\begin{lemma}\label{lem:oliveira}
Let $\bx^1,\bx^2,\ldots,\bx^c$ be independent identically distributed copies of a $d$-dimensional random vector $\bx$ with
\[\TNorm{\bx}\leq M \qquad \mbox{and} \qquad \TNorm{\Expect{\bx\bx^T}}\leq 1.\]
Then, for any $\alpha > 0$,
\[\TNorm{\frac{1}{c}\sum_{i=1}^c \bx^i {\bx^i}^T-\Expect{\bx\bx^T}}\leq \alpha\]
%
holds with probability at least
\[1-\left(2c^2\right)\exp\left(-\frac{c\alpha^2}{16M^2+8M^2\alpha}\right).\]
\end{lemma}
This inequality essentially bounds the probability that the matrix $\frac{1}{c}\sum_{i=1}^c \bx^i {\bx^i}^T$ deviates significantly from its expectation. This deviation is measured with respect to the two norm (namely the largest singular value) of the error matrix.

 \begin{proof}(of Theorem~\ref{thm:theorem7correct})
 	Define the random \textit{row} vector $\by \in \Rs{d}$ as
 	\begin{equation*}
 	\Probab{\by = \frac{1}{\sqrt{p_k}}\bU_{k*}} = p_k \ge \frac{\beta\TNormS{\bU_{k*}}}{d},
 	\end{equation*}
 	for $k=1 ,\ldots, n$. In words, $\by$ is set to be the (rescaled) $k$-th row of $\bU$ with probability $p_k$. Thus, the matrix $\bR$
  has rows $\frac{1}{\sqrt{c}}\by^1,\frac{1}{\sqrt{c}}\by^2,\ldots,\frac{1}{\sqrt{c}}\by^c$, where $\by^1,\by^2,\ldots,\by^c$ are $c$ independent copies of $\by$. Using this notation, it follows that
 	\begin{equation}\label{eqn:expectyyt}
 	\Expect{\by^T\by} = \sum_{k=1}^{n} p_k (\frac{1}{\sqrt{p_k}}\bU_{k*}^T)(\frac{1}{\sqrt{p_k}}\bU_{k*})
 	= \bU^T\bU = \bI_d.
 	\end{equation}
 	Also,
 	\begin{equation*}
 	\bR^T\bR = \frac{1}{c}\sum_{t=1}^c \underbrace{{\by^t}^T \by^t}_{\mathbb{R}^{d \times d}}.
 	\end{equation*}
 	For this vector $\by$, let
 	\begin{equation}\label{eqn:defM}
 	M \ge\TNorm{\by} = \frac{1}{\sqrt{p_k}}\TNorm{\bU_{k*}}.
 	\end{equation}
 	Notice that from Eqn.~(\ref{eqn:expectyyt}) we immediately get $\TNorm{\Expect{\by^T\by}} = \TNorm{\bI_d} = 1$. Applying Lemma~\ref{lem:oliveira} (with $\bx = \by^T$), we get
 	\begin{equation}\label{eqn:ExpectBound}
 	\TNorm{\bR^T\bR-\bU^T\bU} < \epsilon,
 	\end{equation}
 	with probability at least $1-\left(2c\right)^2 \exp\left(-\frac{c\epsilon^2}{16M^2 + 8M^2 \epsilon}\right)$. Let $\delta$ be the failure probability of Theorem~\ref{thm:theorem7correct}; we seek an appropriate value of $c$ in order to guarantee $\left(2c\right)^2 \exp\left(-\frac{c\epsilon^2}{16M^2 + 8M^2 \epsilon}\right) \leq \delta$. Equivalently, we need to satisfy
 	\[\frac{c}{\ln \left(2c/\sqrt{\delta}\right)} \geq \frac{2}{\epsilon^2}\left(16M^2 + 8M^2\epsilon\right).\]
 	Combine Eqns.~(\ref{eqn:defM}) and~(\ref{eqn:appopt4}) to get $M^2 \leq \FNormS{\bU}/\beta=d/\beta$. Recall that $\epsilon < 1$ to conclude that it suffices to choose a value of $c$ such that
 	\[\frac{c}{\ln \left(2c/\sqrt{\delta}\right)} \geq \frac{48d}{\beta\epsilon^2},\]
 	or, equivalently,
 	\[\frac{2c/\sqrt{\delta}}{\ln \left(2c/\sqrt{\delta}\right)} \geq \frac{96d}{\beta\epsilon^2\sqrt{\delta}}.\]
 	We now use the fact that for any $\eta \geq 4$, if $x \geq 2\eta \ln \eta$, then $x/\ln x \geq \eta$. Let $x = 2c/\sqrt{\delta}$, let $\eta = 96d /\left(\beta \epsilon^2\sqrt{\delta}\right)$, and note that $\eta \geq 4$ since $d\geq 1$ and $\beta$, $\epsilon$, and $\delta$ are at most one. Thus, it suffices to set
 	\[\frac{2c}{\sqrt{\delta}} \geq 2 \frac{96 d}{\beta \epsilon^2\sqrt{\delta}}\ln \left(\frac{96 d}{\beta \epsilon^2\sqrt{\delta}}\right),\]
 	which concludes the proof of the theorem.
 \end{proof}

\begin{Remark} Let $\delta = 1/10$ and let $\epsilon$ and $\beta$ be constants. Then, we can compare the bound of Eqn.~(\ref{eqn:cval1}) with the bound of Eqn.~(\ref{eqn:CboundAppendix}) of Theorem~\ref{thm:theorem7correct}: both values of $c$ guarantee the same accuracy $\epsilon$ and the same success probability (say 9/10). However, asymptotically, the bound of Theorem~\ref{thm:theorem7correct} holds by setting $c = O(d\ln d)$, while the bound of Eqn.~(\ref{eqn:cval1}) holds by setting $c=O(d^2)$. Thus, the bound of Theorem~\ref{thm:theorem7correct} is much better. By the Coupon Collector Problem (see Section~\ref{sxn:couponcollector}), sampling-based approaches necessitate at least $\Omega(d\ln d)$ samples, thus making our algorithm asymptotically optimal. 
We should note, however, that deterministic methods exist (see~\cite{Srivastava2010}) that achieve the same bound with $c=O(d/\epsilon^2)$ samples.
\end{Remark}

\begin{Remark}We made no effort to optimize the constants in the expression for $c$ in Eqn.~(\ref{eqn:CboundAppendix}). Better constants are known, by using tighter matrix-Bernstein inequalities. For a state-of-the-art bound see, for example,~\cite[Theorem 5.1]{Holodnak2015}.
\end{Remark}

\subsection{References.} Our presentation in this chapter follows closely the derivations in~\cite{dkm_matrix1}; see~\cite{dkm_matrix1} for a detailed discussion of prior work on this topic. We also refer the interested reader to~\cite{Holodnak2015} and references therein for more recent work on randomized matrix multiplication.

\section{RandNLA Approaches for Regression Problems}\label{sxn:main:regression}

In this section, we will present a simple randomized algorithm for least-squares regression.
In many applications in mathematics and statistical data analysis, it is of interest to find an approximate solution to a system of linear equations that has no exact solution. For example, let a
matrix $\bA \in \Rs{n \times d}$ and a vector $\bvb \in \Rs{n}$ be given. If $n \gg d$, there will not in general exist a vector $\bx \in \Rs{d}$ such that $\bA \bx=\bvb$, and yet it is
often of interest to find a vector $\bx$ such that $\bA \bx \approx \bvb$ in some precise sense. The method of least squares, whose original formulation is often credited to Gauss and Legendre, accomplishes this by minimizing the sum of squares of the elements of the residual vector, i.e., by solving the optimization problem
\begin{equation}
\label{eqn:orig_ls_prob}
\mathcal{Z} = \min_{\bx \in \Rs{d}}
\TNorm{\bA \bx - \bvb}.
\end{equation}
The minimum $\ell_2$-norm vector among those satisfying Eqn.~(\ref{eqn:orig_ls_prob}) is
\begin{equation}
\label{eqn:xopt_orig_ls_prob}
\bx_{opt} = \bA^{\dagger}\bvb,
\end{equation}
where $\bA^{\dagger}$ denotes the Moore-Penrose generalized inverse of the matrix $\bA$. This solution vector has a very natural statistical interpretation as providing an optimal estimator among all linear unbiased estimators, and it has a very natural geometric interpretation as providing an orthogonal projection of the vector $\bvb$ onto the span of the columns of the matrix $\bA$.

Recall that to minimize the quantity in Eqn.~(\ref{eqn:orig_ls_prob}), we can set the derivative of $\TNormS{\bA \bx-\bvb}=(\bA \bx-\bvb)^T(\bA \bx-\bvb)$ with respect to $\bx$ equal to zero, from which it follows that the minimizing vector $\bx_{opt}$ is a solution of the so-called normal equations
\begin{equation}
\label{eqn:normal_eqn}
\bA^T \bA \bx_{opt}=\bA^T\bvb  .
\end{equation}
Computing $\bA^T\bA$, and thus computing $\bx_{opt}$ in this way, takes $O(nd^2)$ time, assuming $n \ge d$.
Geometrically, Eqn.~(\ref{eqn:normal_eqn}) means that the residual vector $\bvb^{\perp}=\bvb-\bA \bx_{opt}$ is required to be orthogonal to the column space of $\bA$, i.e., ${\bvb^{\perp}}^T\bA=0$. While solving the normal equations squares the condition number of the input matrix (and thus is typically not recommended in practice), direct methods (such as the QR decomposition, see Section~\ref{sxn:labasics}) also solve the problem of Eqn.~(\ref{eqn:orig_ls_prob}) in $O(nd^2)$ time, assuming that $n \geq d$. Finally, an alternative expression for the vector $\bx_{opt}$ of Eqn.~(\ref{eqn:xopt_orig_ls_prob}) emerges by leveraging the SVD of $\bA$. If $\bA = \bU_A\bSigma_A \bV_A^T$ denotes the SVD of $\bA$, then
\begin{equation*}
\bx_{opt}=\bV_A\bSigma_A^{-1}\bU_A^T\bvb=\bA^{\dagger}\bvb.
\end{equation*}
Computing $\bx_{opt}$ in this way also takes $O(nd^2)$ time, again assuming $n \ge d$.
In this section, we will describe a randomized algorithm that will provide accurate relative-error approximations to the minimal $\ell_2$-norm solution vector $\bx_{opt}$ of Eqn.~(\ref{eqn:xopt_orig_ls_prob}) faster than these ``exact'' algorithms for a large class of over-constrained least-squares problems.

\subsection{The Randomized Hadamard Transform.}\label{sxn:RHT}
The Randomized Hadamard Transform was introduced in~\cite{Ailon2009} as one step in the development of a fast version of the Johnson-Lindenstrauss lemma. Recall that the $n \times n$ Hadamard matrix (assuming $n$ is a power of two) $\tilde{\bH}_n$, may be defined recursively as follows:
\[ \tilde{\bH}_n = \left[
\begin{array}{cc}
  \widetilde{\bH}_{n/2} & \widetilde{\bH}_{n/2} \\
  \widetilde{\bH}_{n/2} & -\widetilde{\bH}_{n/2}
\end{array}\right]   ,
\qquad \mbox{with} \qquad
\widetilde{\bH}_2 = \left[
\begin{array}{cc}
  +1 & +1 \\
  +1 & -1
\end{array}\right].
\]
We can now define the \textit{normalized} Hadamard transform $\bH_n$ as $(1/\sqrt{n})\tilde{\bH}_n$; it is easy to see that $\bH_n\bH_n^T=\bH_n^T\bH_n=\bI_n$.
Now consider a diagonal matrix $\bD \in \mathbb{R}^{n \times n}$ such that $\bD_{ii}$ is set to +1 with probability $1/2$ and to $-1$ with probability $1/2$. The product $\bH\bD$ is the \emph{Randomized Hadamard Transform} and has three useful properties.
First, when applied to a vector, it ``spreads out'' the mass/energy of that vector, in the sense of providing a bound for the largest element, or infinity norm, of the transformed vector.
Second, computing the product $\bH\bD\bx$ for any vector $\bx \in \mathbb{R}^n$ takes $O(n\log_2 n)$ time. Even better, if we only need to access, say, $r$ elements in the transformed vector, then those $r$ elements can be computed in $O(n \log_2 r )$ time. We will expand on the latter observation in Section~\ref{sxn:lsruntime}, where we will discuss the running time of the proposed algorithm. Third, the Randomized Hadamard Transform is an orthogonal transformation, since $\bH\bD\bD^T\bH^T=\bH^T\bD^T\bD\bH=\bI_n$.

\subsection{The main algorithm and main theorem.} \label{sxn:sampling:result1}
We are now ready to provide an overview of the \textsc{RandLeastSquares} algorithm (Algorithm~\ref{alg:alg_sample_fast}). Let the matrix product $\bH \bD$ denote the $n \times n$ Randomized Hadamard Transform discussed in the previous section.
(For simplicity, we restrict our discussion to the case that $n$ is a power of two, although this restriction can easily be removed by using variants of the Randomized Hadamard Transform~\cite{Mah-mat-rev_BOOK}.)
Our algorithm is a \emph{preconditioned random sampling algorithm}: after premultiplying $\bA$ and $\bvb$ by $\bH \bD$, our algorithm samples uniformly at random $r$ constraints from the preprocessed problem. (See Eqn.~(\ref{eqn:rvaluefinal}), as well as the remarks after Theorem~\ref{thm:alg_sample_fast} for the precise value of $r$.) Then, this algorithm solves the least squares problem on just those sampled constraints to obtain a vector $\tilde{\bx}_{opt} \in \Rs{d}$ such that Theorem~\ref{thm:alg_sample_fast} is satisfied.

\input{alg_Sample_Fast}

Formally, we will let $\bS \in \Rs{n \times r}$ denote a sampling-and-rescaling matrix specifying which of the $n$ (preprocessed) constraints are to be sampled and how they are to be rescaled. This matrix is initially empty and is constructed as described in the \textsc{RandLeastSquares} algorithm.
(We are describing this algorithm in terms of the matrix $\bS$, but as with the \textsc{RandMatrixMultiply} algorithm, we do not need to construct it explicitly in an actual implementation~\cite{AMT10}.)
Then, we can consider the problem
\begin{equation*}
\tilde{\mathcal{Z}}
   = \min_{\bx \in \Rs{d}} \TNorm{\bS^T \bH \bD \bA \bx- \bS^T\bH \bD \bvb } ,
\end{equation*}
which is a least squares approximation problem involving only the $r$ constraints, where the $r$ constraints are uniformly sampled from the matrix $\bA$ after the preprocessing with the Randomized Hadamard Transform. The minimum $\ell_2$-norm vector $\tilde{\bx}_{opt} \in \Rs{d}$ among those that achieve the minimum value $\tilde{\mathcal{Z}}$ in this problem~is
\begin{equation*}
\tilde{\bx}_{opt}
   = \left(\bS^T \bH \bD \bA \right)^{\dagger}\bS^T \bH \bD \bvb     ,
\end{equation*}
which is the output of the \textsc{RandLeastSquares} algorithm.
One can prove (and the proof is provided below) the following theorem about this algorithm.
\begin{theorem}
\label{thm:alg_sample_fast}
Suppose $\bA \in \Rs{n \times d}$ is a matrix of rank $d$, with $n$ being a power of two. Let $\bvb \in \Rs{n}$ and let $\epsilon \in (0,1)$. Run the \textsc{RandLeastSquares} algorithm with
\begin{equation}\label{eqn:rvaluefinal}
r = \max\left\{48^2 d \ln\left(40nd\right)\ln\left(100^2d \ln \left(40nd\right)\right),
40d\ln(40nd)/\epsilon\right\}
\end{equation}
and return $\tilde{\bx}_{opt}$. Then, with probability at least .8, the following two claims hold: first, $\tilde{\bx}_{opt}$ satisfies
\[
\TNorm{\bA \tilde{\bx}_{opt}-\bvb} \le (1+\epsilon) \mathcal{Z} ,
\]
where, recall, that $ \mathcal{Z} $ is given in Eqn.~(\ref{eqn:orig_ls_prob});
and, second, if we assume that $\TNorm{\bU_A  \bU_A^T \bvb} \ge \gamma \TNorm{\bvb}$ for some $\gamma \in (0,1]$, then $\tilde{\bx}_{opt}$ satisfies
\[
\TNorm{\bx_{opt}-\tilde{\bx}_{opt}}
  \leq \sqrt{\epsilon}\left(\kappa(\bA)\sqrt{\gamma^{-2}-1}\right)\TNorm{\bx_{opt}}.
\]
Finally,
\[ n(d+1) + 2n(d+1) \log_2 \left(r + 1\right) + O(rd^2)\]
time suffices to compute the solution $\tilde{\bx}_{opt}$.
\end{theorem}

It is worth noting that the claims of Theorem~\ref{thm:alg_sample_fast} can be made to hold with probability $1-\delta$, for any $\delta>0$, by repeating the algorithm $\left\lceil \ln(1/\delta)/\ln(5)\right\rceil$ times. Also, we note that if $n$ is not a power of two we can pad $\bA$ and $\bvb$ with all-zero rows in order to satisfy the assumption; this process at most doubles the size of the input matrix.

\begin{Remark}
Assuming that $d \leq n \leq e^d$, and using $\max\{a_1,a_2\} \leq a_1 + a_2$, we get that \[r = \mathcal{O} \left( d(\ln d)(\ln n) + \frac{d \ln n}{\epsilon} \right).\] Thus, the running time of the \textsc{RandLeastSquares} algorithm becomes
\[\mathcal{O} \left( nd\ln \frac{d}{\epsilon} + d^3 (\ln d)(\ln n) + \frac{d^3 \ln n}{\epsilon} \right) .\]
Assuming that $n/\ln n = \Omega(d^2)$, the above running time reduces to \[\mathcal{O}\left(nd \ln \frac{d}{\epsilon} + \frac{nd \ln d}{\epsilon}\right).\]
For fixed $\epsilon$, these improve the standard $O(nd^2)$ running time of traditional deterministic algorithms.
It is worth noting that improvements over the standard $O(nd^2)$ time could be derived with weaker assumptions on $n$ and $d$. However, for the sake of clarity of presentation, we only focus on the above setting.
\end{Remark}

\begin{Remark}
The matrix $\bS^T \bH \bD$ can be viewed in one of two equivalent ways:
as a random preprocessing or random preconditioning, which ``uniformizes'' the leverage scores of the input matrix $\bA$ (see Lemma~\ref{lem:HU} for a precise statement), followed by a uniform sampling operation; or as a Johnson-Lindenstrauss style random projection, which preserves the geometry of the entire span of $\bA$, rather than just a discrete set of points (see Lemma~\ref{lem:sample_lem20pf} for a precise statement).
\end{Remark}

\subsection{RandNLA algorithms as preconditioners.}\label{sxn:precond}
Stepping back, recall that the \textsc{RandLeastSquares} algorithm may be viewed as preconditioning the input matrix $\bA$ and the target vector $\bvb$ with a carefully-constructed data-independent random matrix $\bX$. 
(Since the analysis of the \textsc{RandLowRank} algorithm, our main algorithm for low-rank matrix approximation, in Section~\ref{sxn:main:lowrank} below, boils down to very similar ideas as the analysis of the \textsc{RandLeastSquares} algorithm, the ideas underlying the following discussion also apply to the \textsc{RandLowRank} algorithm.)
For our random sampling algorithm, we let $\bX = \bS^T \bH \bD$, where $\bS$ is a matrix that represents the sampling operation and $\bH \bD$ is the Randomized Hadamard Transform. Thus, we replace the least squares approximation problem of Eqn.~(\ref{eqn:orig_ls_prob}) with the least squares approximation problem
\begin{equation}
\label{eqn:orig_ls_prob_Xrotated}
\tilde{\mathcal{Z}} = \min_{\bx \in \Rs{d}} \TNorm{\bX (\bA \bx - \bvb)}.
\end{equation}
We explicitly compute the solution to the above problem using a traditional deterministic algorithm, e.g., by computing the vector
\begin{equation}
\label{eqn:xopt_orig_ls_prob_Xrotated}
\tilde{\bx}_{opt} = \left(\bX \bA\right)^{\dagger}\bX \bvb  .
\end{equation}
Alternatively, one could use standard iterative methods such as the the Conjugate Gradient Normal Residual method, which can produce an $\epsilon$-approximation to the optimal solution of Eqn.~(\ref{eqn:orig_ls_prob_Xrotated}) in $O(\kappa(\bX \bA) rd \ln(1/\epsilon))$ time, where $\kappa(\bX \bA)$ is the condition number of $\bX \bA$ and $r$ is the number of rows of $\bX \bA$. This was indeed the strategy implemented in the popular Blendenpik/LSRN approach~\cite{AMT10}.

We now state and prove a lemma that establishes sufficient conditions on \emph{any} matrix $\bX$ such that the solution vector $\tilde{\bx}_{opt}$ to the least squares problem of Eqn.~(\ref{eqn:orig_ls_prob_Xrotated}) will satisfy the relative-error bounds of Theorem~\ref{thm:alg_sample_fast}. Recall that the SVD of $\bA$ is $\bA=\bU_A \bSigma_A \bV_A^T$. In addition, for notational simplicity, we let $\bvb^{\perp} = \bU_A^{\perp}{\bU_A^{\perp}}^{T}\bvb$ denote the  part of the right hand side vector $\bvb$ lying outside of the column space of $\bA$.

The two conditions that we will require for the matrix $\bX$ are:
\begin{eqnarray}
\label{eqn:lemma1_ass1}
& & \sigma_{min}^2 \left( \bX \bU_A \right) \ge 1/\sqrt{2} \mbox{; and}  \\
\label{eqn:lemma1_ass2} & &
\TNormS{\bU_A^T \bX^T \bX \bvb^{\perp}}
      \le \epsilon \mathcal{Z}^2/2  ,
\end{eqnarray}
for some $\epsilon \in (0,1)$.
Several things should be noted about these conditions.
\begin{itemize}
\item
First, although Condition~(\ref{eqn:lemma1_ass1}) only states that $\sigma_i^2(\bX \bU_A)\geq 1/\sqrt{2}$, for all $i =1,\ldots, d$, our randomized algorithm satisfies $\abs{1-\sigma_i^2(\bX \bU_A)} \le 1-1/\sqrt{2}$, for all $i =1,\ldots, d$.
This is equivalent to $$\TNorm{I-\bU_A^T\bX^T \bX \bU_A } \le 1-1/\sqrt{2}.$$
Thus, one should think of $\bX \bU_A$ as an approximate isometry.
\item
Second, the lemma is a deterministic statement, since it makes no explicit reference to a particular randomized algorithm and since $\bX$ is not assumed to be constructed from a randomized process. Failure probabilities will enter later when we show that our randomized algorithm constructs an $\bX$ that satisfies Conditions~(\ref{eqn:lemma1_ass1}) and~(\ref{eqn:lemma1_ass2}) with some probability.
\item
Third, Conditions~(\ref{eqn:lemma1_ass1}) and~(\ref{eqn:lemma1_ass2}) define what has come to be known as a \emph{subspace embedding}, since it is an embedding that preserves the geometry of the entire subspace of the matrix $\bA$.
Such a subspace embedding can be \emph{oblivious} (meaning that it is constructed without knowledge of the input matrix, as with random projection algorithms) or \emph{non-oblivious} (meaning that it is constructed from information in the input matrix, as with data-dependent nonuniform sampling algorithms).
This style of analysis represented a major advance in RandNLA algorithms, since it premitted much stronger bounds to be obtained than had been possible with previous methods.
See~\cite{DMMS11} for the journal version (which was a combination and extension of two previous conference papers) of the first paper to use this style of analysis.
\item
Fourth, Condition~(\ref{eqn:lemma1_ass2}) simply states that $\bX \bvb^{\perp}=\bX \bU_A^{\perp}{\bU_A^{\perp}}^{T}\bvb$ remains approximately orthogonal to $\bX \bU_A$.  Clearly, before applying $\bX$, it holds that $\bU_A^T\bvb^{\perp}=0$.
\item
Fifth, although Condition~(\ref{eqn:lemma1_ass2}) depends on the right hand side vector $\bvb$, the \textsc{RandLeastSquares} algorithm will satisfy it without using any information from $\bvb$.
(See Lemma~\ref{lem:sample_lem40pf} below.)
\end{itemize}
Given Conditions~(\ref{eqn:lemma1_ass1}) and~(\ref{eqn:lemma1_ass2}), we can establish the following lemma.

\begin{lemma} \label{lem:suff_cond}
Consider the overconstrained least squares approximation problem of Eqn.~(\ref{eqn:orig_ls_prob}) and let the matrix $\bU_A \in \Rs{n \times d}$ contain the top $d$ left singular vectors of $\bA$. Assume that the matrix $\bX$ satisfies Conditions~(\ref{eqn:lemma1_ass1}) and~(\ref{eqn:lemma1_ass2}) above, for some $\epsilon \in (0,1)$. Then, the solution vector $\tilde{\bx}_{opt}$ to the least squares approximation
problem~(\ref{eqn:orig_ls_prob_Xrotated}) satisfies:
\begin{eqnarray}
\label{eqn:lemma1_eq3}
\TNorm{\bA\tilde{\bx}_{opt}-\bvb} &\le& (1+\epsilon) \mathcal{Z}  \mbox{, and}  \\
\label{eqn:lemma1_eq4}
\TNorm{\bx_{opt}-\tilde{\bx}_{opt}}
  &\leq& \frac{1}{\sigma_{min}(\bA)}\sqrt{\epsilon}\mathcal{Z}  .
\end{eqnarray}
\end{lemma}
\begin{proof}
Let us first rewrite the down-scaled regression problem induced by
$\bX$ as
\begin{eqnarray}
\nonumber \min_{\bx \in \Rs{d}} \TNormS{ \bX \bvb - \bX \bA \bx} &=& \min_{\bx \in \Rs{d}} \TNormS{\bX \bA x-\bX \bvb}\\
\label{eqn:ds1} &=& \min_{\by \in \Rs{d}} \TNormS{
\bX \bA(\bx_{opt}+\by)-\bX( \bA x_{opt}+\bvb^{\perp})}                 \\
\nonumber
  &=& \min_{\by \in \Rs{d}} \TNormS{\bX \bA \by-\bX \bvb^{\perp}}    \\
\label{eqn:ds2}
  &=& \min_{\bz \in \Rs{d}} \TNormS{\bX \bU_A\bz-\bX \bvb^{\perp}}.
\end{eqnarray}
Eqn.~(\ref{eqn:ds1}) follows since $\bvb=\bA \bx_{opt}+\bvb^{\perp}$ and Eqn.~(\ref{eqn:ds2}) follows since the columns of the matrix $\bA$ span the same subspace as the columns of $\bU_A$. Now, let $\bz_{opt} \in \Rs{d}$  be such that
$\bU_A \bz_{opt} = \bA (\tilde{\bx}_{opt}-\bx_{opt})$. Using this value for $\bz_{opt}$, we will prove that $\bz_{opt}$ is minimizer of the above optimization problem, as follows:
\begin{eqnarray}
\nonumber\TNormS{\bX \bU_A \bz_{opt} - \bX \bvb^{\perp}} &=&   \TNormS{\bX \bA (\tilde{\bx}_{opt}-\bx_{opt}) - \bX \bvb^{\perp}}\\
\nonumber&=&   \TNormS{\bX \bA \tilde{\bx}_{opt} - \bX \bA \bx_{opt} - \bX \bvb^{\perp}}\\
\label{eqn:pdch51}&=&   \TNormS{\bX \bA \tilde{\bx}_{opt} - \bX\bvb}\\
\nonumber&=&  \min_{\bx \in \Rs{d}} \TNormS{\bX \bA \bx-\bX \bvb}  \\ 
\nonumber &=& \min_{\bz \in \Rs{d}} \TNormS{\bX \bU_Az-\bX \bvb^{\perp}}.
\end{eqnarray}
Eqn.~(\ref{eqn:pdch51}) follows since $\bvb=\bA \bx_{opt}+\bvb^{\perp}$ and the last equality follows from Eqn.~(\ref{eqn:ds2}). Thus, by the normal equations~(\ref{eqn:normal_eqn}), we have that
\begin{equation*}
\label{eqn:ds-normal} (\bX \bU_A)^T\bX \bU_A \bz_{opt} =
(\bX \bU_A)^T \bX \bvb^{\perp}.
\end{equation*}
Taking the norm of both sides and observing that under Condition~(\ref{eqn:lemma1_ass1}) we have $\sigma_i((\bX \bU_A)^T \bX \bU_A) = \sigma_i^2(\bX \bU_A) \ge 1/\sqrt{2}$, for all $i$, it follows that
\begin{equation}
\label{eqn:z-norm1}
  \TNormS{\bz_{opt}} / 2  \le \TNormS{(\bX \bU_A)^T\bX \bU_Az_{opt}} = \TNormS{
(\bX \bU_A)^T \bX \bvb^{\perp} }.
\end{equation}
Using Condition~(\ref{eqn:lemma1_ass2}) we observe that
\begin{equation}
\label{eqn:z-norm2}
   \TNormS {\bz_{opt}} \le \epsilon\mathcal{Z}^2.
\end{equation}

To establish the first claim of the lemma, let us rewrite
the norm of the residual vector as
\begin{eqnarray}
\TNormS{ \bvb - \bA \tilde{\bx}_{opt} } \nonumber
   &=& \TNormS{ \bvb - \bA \bx_{opt} + \bA \bx_{opt} - \bA \tilde{\bx}_{opt} }  \\
\label{eqn:pfCeq1}
   &=& \TNormS{ \bvb - \bA \bx_{opt} } + \TNormS{ \bA \bx_{opt} - \bA \tilde{\bx}_{opt} } \\
\label{eqn:pfCeq2}
   &=& \mathcal{Z}^{2} + \TNormS{-\bU_A \bz_{opt}} \\
\label{eqn:pfCeq3}
   &\leq& \mathcal{Z}^{2} + \epsilon \mathcal{Z}^{2} ,
\end{eqnarray}
where Eqn.~(\ref{eqn:pfCeq1}) follows by the Pythagorean theorem, since $\bvb - \bA \bx_{opt}
= \bvb^\perp$ is orthogonal to $\bA$ and consequently to
$\bA(\bx_{opt} - \tilde{\bx}_{opt})$; Eqn.~(\ref{eqn:pfCeq2}) follows
by the definition of $\bz_{opt}$ and $\mathcal{Z}$; and Eqn.~(\ref{eqn:pfCeq3})
follows by (\ref{eqn:z-norm2}) and fact that $U_A$ has orthonormal columns. The first claim of the lemma follows since
$\sqrt{1+\epsilon} \le 1+\epsilon$.

To establish the second claim of the lemma, recall that $\bA(\bx_{opt}-\tilde{\bx}_{opt})=\bU_A\bz_{opt}$. If we take the norm of both sides of this expression, we have that
\begin{eqnarray}
\TNormS{ \bx_{opt}-\tilde{\bx}_{opt} } \label{eqn:pfDeq1}
  &\leq& \frac {\TNormS{\bU_Az_{opt}}} {\sigma_{min}^2(\bA)} \\
\label{eqn:pfDeq2}
  &\leq& \frac {\epsilon\mathcal{Z}^2} {\sigma_{min}^2(\bA)},
\end{eqnarray}
where Eqn.~(\ref{eqn:pfDeq1}) follows since $\sigma_{min}(\bA)$ is the
smallest singular value of $\bA$ and since the rank of $\bA$ is
$d$; and Eqn.~(\ref{eqn:pfDeq2}) follows by
Eqn.~(\ref{eqn:z-norm2}) and the orthonormality of the columns of $\bU_A$. Taking the square root, the second claim of the
lemma follows.
\end{proof}

If we make no assumption on $\bvb$, then Eqn.~(\ref{eqn:lemma1_eq4}) from
Lemma~\ref{lem:suff_cond} may provide a weak bound in terms of
$\TNorm{\bx_{opt}}$. If, on the other hand, we make the additional
assumption that a constant fraction of the norm of $\bvb$ lies in
the subspace spanned by the columns of $\bA$,
then Eqn.~(\ref{eqn:lemma1_eq4}) can be strengthened.
Such an assumption is reasonable, since most least-squares problems are
practically interesting if at least some part of $\bvb$ lies in the subspace
spanned by the columns of $\bA$.

\begin{lemma}
\label{lem:suff_cond2}
Using the notation of Lemma~\ref{lem:suff_cond}, and additionally assuming
that $\TNorm{\bU_A \bU_A^T\bvb} \geq \gamma\TNorm{\bvb}$, for some fixed
$\gamma \in (0,1]$, it follows that
\begin{equation}
\TNorm{\bx_{opt}-\tilde{\bx}_{opt}}
  \leq \sqrt{\epsilon}\left(\kappa(\bA)\sqrt{\gamma^{-2}-1}\right)\TNorm{\bx_{opt}} .
\end{equation}
\end{lemma}

\begin{proof}
Since $\TNorm{\bU_A \bU_A^T\bvb} \geq \gamma\TNorm{\bvb}$, it follows that
\begin{eqnarray}
         \mathcal{Z}^2
\nonumber    &=&    \TNormS{\bvb} - \TNormS{\bU_A \bU_A^T \bvb}          \\
\nonumber    &\leq& (\gamma^{-2}-1) \TNormS{\bU_A \bU_A^T \bvb}         \\
\nonumber    &\leq&
{\sigma_{\max}^{2}(\bA)}(\gamma^{-2}-1)\TNormS{\bx_{opt}}  .
\end{eqnarray}
This last inequality follows from $\bU_A \bU_A^T\bvb = \bA \bx_{opt}$, which implies \[ \TNorm{\bU_A \bU_A^T \bvb} = \TNorm{\bA \bx_{opt}} \leq \TNorm{\bA} \TNorm{\bx_{opt}} = \sigma_{\max}\left(\bA\right)\TNorm{\bx_{opt}}. \]
By combining this with Eqn. (\ref{eqn:lemma1_eq4}) of
Lemma~\ref{lem:suff_cond}, the lemma follows.
\end{proof}


\subsection{The proof of Theorem~\ref{thm:alg_sample_fast}.}\label{sxn:ls:thmproof}

To prove Theorem~\ref{thm:alg_sample_fast}, we adopt the following approach: we first show that the Randomized Hadamard Transform has the effect preprocessing or preconditioning the input matrix to make the leverage scores approximately uniform; and we then show that Condition~(\ref{eqn:lemma1_ass1}) and~(\ref{eqn:lemma1_ass2}) can be satisfied by sampling uniformly on the preconditioned input.
The theorem will then follow from Lemma~\ref{lem:suff_cond}.

\subsubsection*{The effect of the Randomized Hadamard Transform.}
We start by stating a lemma that quantifies the manner in which $\bH \bD$ approximately ``uniformizes'' information in the left singular subspace of the matrix $\bA$; this will allow us to sample uniformly and apply our randomized matrix multiplication results from Section~\ref{chapter:MM} in order to analyze the proposed algorithm. We state the lemma for a general $n \times d$ orthogonal matrix $\bU$ such that $\bU^T \bU = \bI_d$.

\begin{lemma} \label{lem:HU}
Let $\bU$ be an $n \times d$ orthogonal matrix and let the product $\bH \bD$ be the $n \times n$ Randomized Hadamard Transform of Section~\ref{sxn:RHT}. Then, with probability at least $.95$,
\begin{eqnarray}
\label{eqn:lem:HU_eqn2} \TNormS{\left(\bH \bD \bU \right)_{i*}}
   &\leq& \frac{2d\ln(40nd)}{n},\qquad
       \text{ for all } i =1,\ldots, n   .
\end{eqnarray}
\end{lemma}
The following well-known inequality~\cite[Theorem 2]{Hoeffding1963} will be useful in the proof.
(See also the chapter by Vershynin in this volume~\cite{pcmi-chapter-vershynin} for related results.)
\begin{lemma}\label{lem:hoef}
Let $X_i$, $i=1,\ldots, n$ be independent random variables with finite first and second moments such that, for all $i$, $a_i \leq X_i \leq b_i$.  Then, for any $t >0$,
\[\Probab{\abs{\sum_{i=1}^n X_i - \sum_{i=1}^n \Expect{X_i}} \geq nt} \leq 2\exp{\left(-\frac{2n^2t^2}{\sum_{i=1}^{n}(a_i - b_i)^2}\right)}.\]
\end{lemma}

Given this lemma, we now provide the proof of Lemma~\ref{lem:HU}.

\begin{proof}(of Lemma~\ref{lem:HU})
Consider $\left(\bH \bD \bU \right)_{ij}$ for some $i$, $j$ (recalling that $i=1,\ldots, n$ and $j = 1,\ldots, d$). Recall that $\bD$ is a diagonal matrix; then,
\begin{equation*}
\left(\bH \bD \bU \right)_{ij} = \sum_{\ell=1}^n \bH_{i\ell}\bD_{\ell \ell} \bU_{\ell j} =
\sum_{\ell=1}^n \bD_{\ell \ell} \left(\bH_{i\ell} \bU_{\ell j}\right)=\sum_{\ell=1}^n X_{\ell}.
\end{equation*}
Let $X_{\ell}=\bD_{\ell \ell} \left(\bH_{i\ell} \bU_{\ell j}\right)$ be our set of $n$ (independent) random variables. By the construction of $\bD$ and $\bH$, it is easy to see that $\Expect{X_{\ell}}=0$; also,
\[\abs{X_{\ell}} = \abs{\bD_{\ell \ell} \left(\bH_{i\ell} \bU_{\ell j}\right)}\leq \frac{1}{\sqrt{n}}\abs{\bU_{\ell j}}.\]
Applying Lemma~\ref{lem:hoef}, we get
\[\Probab{\abs{\left(\bH \bD \bU \right)_{ij}} \geq nt}\leq 2\exp{\left(-\frac{2n^3t^2}{4\sum_{\ell=1}^{n}\bU_{\ell j}^2}\right)}=
2\exp{\left(-n^3t^2/2\right)}.\]
In the last equality we used the fact that $\sum_{\ell=1}^{n}\bU_{\ell j}^2=1$, i.e., that the columns of $\bU$ are unit-length. Let the right-hand side of the above inequality be equal to $\delta$ and solve for $t$ to get
\[\Probab{\abs{\left(\bH \bD \bU \right)_{ij}} \geq \sqrt{\frac{2 \ln(2/\delta)}{n}}}\leq \delta.\]
Let $\delta = 1/(20nd)$ and apply the union bound over all $nd$ possible index pairs $(i,j)$ to get that, 
with probability at least 1-1/20=0.95,
for all $i,j$, 
\[\abs{\left(\bH \bD \bU \right)_{ij}} \leq \sqrt{\frac{2 \ln(40nd)}{n}}.\]
Thus,
\begin{equation}\label{eqn:eqP1}
\TNormS{\left(\bH \bD \bU \right)_{i*}} = \sum_{j=1}^d \left(\bH \bD \bU\right)_{ij}^2 \leq \frac{2d\ln(40nd)}{n}
\end{equation}
for all $i =1,\ldots,  n$, which concludes the proof of the lemma.
\end{proof}

\subsubsection*{Satisfying Condition~(\ref{eqn:lemma1_ass1}).}
We next prove the following lemma, which states that all the singular values of $\bS^T \bH \bD \bU_A$ are close to one, and in particular that Condition~(\ref{eqn:lemma1_ass1}) is satisfied by the \textsc{RandLeastSquares} algorithm.
The proof of this Lemma~\ref{lem:sample_lem20pf} essentially follows from our results in Theorem~\ref{thm:theorem7correct} for the \textsc{RandMatrixMultiply} algorithm (for approximating the product of a matrix and its transpose).

\begin{lemma}
\label{lem:sample_lem20pf}
Assume that Eqn.~(\ref{eqn:lem:HU_eqn2}) holds. If
\begin{equation}\label{eqn:rvalue}
r \geq 48^2 d \ln\left(40nd\right)\ln\left(100^2d \ln \left(40nd\right)\right)  ,
\end{equation}
then, with probability at least .95,
\[\abs{1 - \sigma_i^2\left(\bS^T \bH \bD \bU_A\right)} \leq 1-\frac{1}{\sqrt{2}}\]
holds for all $i =1,\ldots, d$.
\end{lemma}
\begin{proof}(of Lemma~\ref{lem:sample_lem20pf})
Note that for all $i =1,\ldots, d$,
\begin{eqnarray}
\abs{1 - \sigma_i^2\left(\bS^T \bH \bD \bU_A\right)} \nonumber
   &=&    \abs{\sigma_i\left(\bU_A^T \bD \bH^T \bH \bD \bU_A\right)
            - \sigma_i\left(\bU_A^T \bD \bH^T \bS \bS^T \bH \bD \bU_A\right)}    \\
\label{eqn:eqX31}
   &\leq& \TNorm{\bU_A^T \bD \bH^T \bH \bD \bU_A - \bU_A^T \bD \bH^T \bS \bS^T \bH \bD \bU_A}.
\end{eqnarray}
In the above, we used the fact that $\bU_A^T \bD \bH^T \bH \bD \bU_A = \bI_d$ and inequality~(\ref{eqn:svineq1}) that was discussed in our Linear Algebra review in Section~\ref{sxn:review:SVD}. We now view $\bU_A^T \bD \bH^T \bS \bS^T\bH \bD \bU_A$ as an approximation to the product of two matrices, namely $\bU_A^T \bD \bH^T=\left(\bH \bD \bU_A\right)^T$ and $\bH \bD \bU_A$, constructed by randomly sampling and rescaling columns of $\left(\bH \bD \bU_A\right)^T$. Thus, we can leverage Theorem~\ref{thm:theorem7correct}.

More specifically, consider the matrix $\left(\bH \bD \bU_A\right)^T$. Obviously, since $\bH$, $\bD$, and $\bU_A$ are orthogonal matrices, $\TNorm{\bH \bD \bU_A}=1$ and $\FNorm{\bH \bD \bU_A}=\FNorm{\bU_A}=\sqrt{d}$. Let $\beta = \left(2\ln(40nd)\right)^{-1}$; since we assumed that Eqn.~(\ref{eqn:lem:HU_eqn2}) holds, we note that the columns of $\left(\bH \bD \bU_A \right)^T$, which correspond to the rows of $\bH \bD \bU_A$, satisfy
\begin{equation}
\label{eqn:unif_prob_OK} \frac{1}{n}
   \ge  \beta \frac{\TNormS{\left(\bH \bD \bU_A\right)_{i*}}}{\FNormS{\bH \bD \bU_A}}, \qquad
       \text{ for all } i =1,\ldots, n  .
\end{equation}
Thus, applying Theorem~\ref{thm:theorem7correct} with $\beta = \left(2\ln(40nd)\right)^{-1}$, $\epsilon = 1 - 1/\sqrt{2}$, and $\delta = 1/20$ implies that
\[\TNorm{\bU_A^T \bD \bH^T \bH \bU_A - \bU_A^T\bD \bH^T \bS \bS^T \bH \bD \bU_A} \leq 1-\frac{1}{\sqrt{2}} \]
holds with probability at least $1-1/20=.95$. For the above bound to hold, we need $r$ to assume the value of Eqn.~(\ref{eqn:rvalue}). Finally, we note that since $\FNormS{\bH \bD \bU_A}=d \geq 1$, the assumption of Theorem~\ref{thm:theorem7correct} on the Frobenius norm of the input matrix is always satisfied. Combining the above with inequality~(\ref{eqn:eqX31}) concludes the proof of the lemma.
\end{proof}

\subsubsection*{Satisfying Condition~(\ref{eqn:lemma1_ass2}).}
We next prove the following lemma, which states that Condition~(\ref{eqn:lemma1_ass2}) is satisfied by the \textsc{RandLeastSquares} algorithm.
The proof of this Lemma~\ref{lem:sample_lem40pf} again essentially follows from our bounds for the \textsc{RandMatrixMultiply} algorithm from Section~\ref{chapter:MM} (except here it is used for approximating the product of a matrix and a vector).

\begin{lemma}
\label{lem:sample_lem40pf}
Assume that Eqn.~(\ref{eqn:lem:HU_eqn2}) holds.  If $r \geq 40d\ln(40nd)/\epsilon$, then, with probability at least .9,
\[\TNormS{
\left(\bS^T \bH \bD \bU_A\right)^{T}\bS^T \bH \bD \bvb^{\perp}
} \leq \epsilon \mathcal{Z}^{2}/2.\]
\end{lemma}
\begin{proof}(of Lemma~\ref{lem:sample_lem40pf})
Recall that $\bvb^{\perp} = \bU_A^{\perp}{\bU_A^{\perp}}^{T}\bvb$ and that ${\mathcal Z} = \TNorm{\bvb^{\perp}}$. We start by noting that since
$\TNormS{\bU_A^T \bD \bH^T \bH \bD \bvb^{\perp}}=\TNormS{\bU_A^T \bvb^{\perp}}=0$ it follows that
\[
 \TNormS{ \left(\bS^T \bH \bD \bU_A\right)^{T}\bS^T \bH \bD \bvb^{\perp} }
   = \TNormS{\bU_A^T \bD \bH^T \bS \bS^T \bH \bD \bvb^{\perp}
             - \bU_A^T \bD \bH^T \bH \bD \bvb^{\perp}}    .
\]
Thus, we can view $\left(\bS^T \bH \bD \bU_A\right)^{T}\bS^T \bH \bD \bvb^{\perp}$ as approximating the product of two matrices, $\left(\bH \bD \bU_A\right)^{T}$ and $\bH \bD \bvb^{\perp}$, by randomly sampling columns from $\left(\bH \bD \bU_A\right)^T$ and rows (elements) from $\bH \bD \bvb^{\perp}$. Note that the sampling probabilities are uniform and do not depend on the norms of the columns of $\left(\bH \bD \bU_A\right)^{T}$ or the rows of $\bH \bvb^{\perp}$. We will apply the bounds of Eqn.~(\ref{eqn:appopt2result}), after arguing that the assumptions of Eqn.~(\ref{eqn:appopt2}) are satisfied. Indeed, since we condition on Eqn.~(\ref{eqn:lem:HU_eqn2}) holding, the rows of $\bH \bD \bU_A$ (which of course correspond to columns of $\left(\bH \bD \bU_A\right)^T$) satisfy
\begin{equation}
\frac{1}{n}   \ge  \beta \frac{\TNormS{\left(\bH \bD \bU_A\right)_{i*}}}{\FNormS{\bH \bD \bU_A}}, \qquad       \text{ for all } i =1,\ldots, n,
\end{equation}
for $\beta = \left(2\ln(40nd)\right)^{-1}$. Thus, Eqn.~(\ref{eqn:appopt2result}) implies
\begin{equation*}
\Expect{ \TNormS{\left(\bS^T \bH \bD \bU_A\right)^{T}\bS^T \bH \bD \bvb^{\perp}
} }
   \leq \frac{1}{\beta r}\FNormS{\bH \bD \bU_A}\TNormS{\bH \bD \bvb^{\perp}}
   =     \frac{d{\mathcal Z}^2}{\beta r}.
\end{equation*}
In the above we used $\FNormS{\bH \bD \bU_A} = d$. Markov's inequality now implies that with probability at least .9,
\begin{equation*}
\TNormS{\left(\bS^T \bH \bD \bU_A\right)^{T}\bS^T \bH \bD \bvb^{\perp}
} \leq \frac{10d{\mathcal Z}^2}{\beta r}.
\end{equation*}
Setting $r \geq 20d/(\beta\epsilon)$ and using the value of $\beta$ specified above concludes the proof of the lemma.
\end{proof}

\subsubsection*{Completing the proof of Theorem~\ref{thm:alg_sample_fast}.}
The theorem follows since Lemmas~\ref{lem:sample_lem20pf} and~\ref{lem:sample_lem40pf} establish that the sufficient conditions of Lemma~\ref{lem:suff_cond} hold.
In more detail, we now complete the proof of Theorem~\ref{thm:alg_sample_fast}. First, let ${\mathcal E}_{(\ref{eqn:lem:HU_eqn2})}$ denote the event that Eqn.~(\ref{eqn:lem:HU_eqn2}) holds; clearly, $\Probab{{\mathcal E}_{(\ref{eqn:lem:HU_eqn2})}} \geq .95$. Second, let ${\mathcal E}_{\ref{lem:sample_lem20pf},\ref{lem:sample_lem40pf}|(\ref{eqn:lem:HU_eqn2})}$ denote the event that both Lemmas~\ref{lem:sample_lem20pf} and~\ref{lem:sample_lem40pf} hold conditioned on ${\mathcal E}_{(\ref{eqn:lem:HU_eqn2})}$ holding. Then,
\begin{align*}
{\mathcal E}_{\ref{lem:sample_lem20pf},\ref{lem:sample_lem40pf}|(\ref{eqn:lem:HU_eqn2})} &= 1 - \overline{{\mathcal E}_{\ref{lem:sample_lem20pf},\ref{lem:sample_lem40pf}|(\ref{eqn:lem:HU_eqn2})}}\\
&= 1 - \bf{Pr}\left(\left(\mbox{Lemma \ref{lem:sample_lem20pf} does not hold} | {\mathcal E}_{(\ref{eqn:lem:HU_eqn2})}\right)\right. \\
&\hspace{15mm} \textbf{OR}\left.\left(\mbox{Lemma \ref{lem:sample_lem40pf} does not hold} | {\mathcal E}_{(\ref{eqn:lem:HU_eqn2})}\right)\right)\\
&\ge 1 - \Probab{\left(\mbox{Lemma \ref{lem:sample_lem20pf} does not hold} | {\mathcal E}_{(\ref{eqn:lem:HU_eqn2})}\right)}\\
&\hspace{15mm} -\Probab{\left(\mbox{Lemma \ref{lem:sample_lem40pf} does not hold} | {\mathcal E}_{(\ref{eqn:lem:HU_eqn2})}\right)}\\
&\geq 1 - .05 - .1 = .85.
\end{align*}
In the above, $\overline{\mathcal E}$ denotes the complement of event ${\mathcal E}$. In the first inequality we used the union bound and in the second inequality we leveraged the bounds for the failure probabilities of Lemmas~\ref{lem:sample_lem20pf} and~\ref{lem:sample_lem40pf}, given that Eqn.~(\ref{eqn:lem:HU_eqn2}) holds. We now let ${\mathcal E}$ denote the event that both Lemmas~\ref{lem:sample_lem20pf} and~\ref{lem:sample_lem40pf} hold, without any a priori conditioning on event ${\mathcal E}_{(\ref{eqn:lem:HU_eqn2})}$; we will bound $\Probab{\mathcal E}$ as follows:
\begin{eqnarray*}
\Probab{\mathcal E} &=& \Probab{{\mathcal E} | {\mathcal E}_{(\ref{eqn:lem:HU_eqn2})}}\cdot \Probab{{\mathcal E}_{(\ref{eqn:lem:HU_eqn2})}}
+\Probab{{\mathcal E} | \overline{{\mathcal E}_{(\ref{eqn:lem:HU_eqn2})}}}\cdot \Probab{\overline{{\mathcal E}_{(\ref{eqn:lem:HU_eqn2})}}}\\
&\geq& \Probab{{\mathcal E} | {\mathcal E}_{(\ref{eqn:lem:HU_eqn2})}}\cdot \Probab{{\mathcal E}_{(\ref{eqn:lem:HU_eqn2})}}\\
&=& \Probab{{\mathcal E}_{\ref{lem:sample_lem20pf},\ref{lem:sample_lem40pf}|(\ref{eqn:lem:HU_eqn2})} | {\mathcal E}_{(\ref{eqn:lem:HU_eqn2})}}\cdot \Probab{{\mathcal E}_{(\ref{eqn:lem:HU_eqn2})}}\\
&\geq& .85\cdot .95 \geq .8.
\end{eqnarray*}
In the first inequality we used the fact that all probabilities are positive. The above derivation immediately bounds the success probability of Theorem~\ref{thm:alg_sample_fast}. Combining Lemmas~\ref{lem:sample_lem20pf} and~\ref{lem:sample_lem40pf} with the structural results of Lemma~\ref{lem:suff_cond} and setting $r$ as in Eqn.~(\ref{eqn:rvaluefinal}) concludes the proof of the accuracy guarantees of Theorem~\ref{thm:alg_sample_fast}.

\subsection{The running time of the {\small R{\scriptsize AND}L{\scriptsize EAST}S{\scriptsize QUARES}} algorithm.}\label{sxn:lsruntime}
We now discuss the running time of the \textsc{RandLeastSquares} algorithm. First of all, by the construction of $\bS$, the number of non-zero entries in $\bS$ is $r$. In Step $6$ we need to compute the products $\bS^T \bH \bD \bA$ and $\bS^T \bH \bD \bvb$. Recall that $\bA$ has $d$ columns and thus the running time of computing both products is equal to the time needed to apply $\bS^T \bH \bD$ on $(d+1)$ vectors.
In order to apply $\bD$ on $(d+1)$ vectors in $\Rs{n}$, $n(d+1)$ operations suffice. In order to estimate how many operations are needed to apply $\bS^T \bH$ on $(d+1)$ vectors, we use the following analysis that was first proposed in~\cite[Section 7]{AL09}.

Let $\bx$ be any vector in $\mathbb{R}^n$; multiplying $\bH$ by $\bx$ can be done as follows:
\begin{align*}
\begin{pmatrix} \bH_{n/2} & \bH_{n/2} \\ \bH_{n/2} & -\bH_{n/2} \end{pmatrix} \begin{pmatrix} \bx_1 \\ \bx_2 \end{pmatrix} = \begin{pmatrix} \bH_{n/2}(\bx_1 + \bx_2) \\ \bH_{n/2}(\bx_1-\bx_2) \end{pmatrix}.
\end{align*}
Let $T(n)$ be the number of operations required to perform this operation for $n$-dimensional vectors. Then,
\[T(n) = 2T(n/2) + n,\]
and thus $T(n) = O(n \log n)$. We can now include the sub-sampling matrix $\bS$ to get
\begin{align*}
\begin{pmatrix} \bS_1 & \bS_2 \end{pmatrix} \begin{pmatrix} \bH_{n/2} & \bH_{n/2} \\ \bH_{n/2} & -\bH_{n/2} \end{pmatrix} \begin{pmatrix} \bx_1 \\ \bx_2 \end{pmatrix} = \bS_1 \bH_{n/2}(\bx_1 + \bx_2) + \bS_2 \bH_{n/2} (\bx_1 - \bx_2).
\end{align*}
Let $\nnz{\cdot}$ denote the number of non-zero entries of its argument. Then,
\[T(n,\nnz{\bS}) = T(n/2, \nnz{\bS_1}) + T(n/2,\nnz{\bS_2}) + n. \]
From standard methods in the analysis of recursive algorithms, we can now use the fact that $r=\nnz{\bS}=\nnz{\bS_1}+\nnz{\bS_2}$ to prove that
\[T(n,r) \leq 2n \log_2(r+1).\]
Towards that end, let $r_1=\nnz{\bS_1}$ and let $r_2=\nnz{\bS_2}$. Then,
\begin{align*}
T(n,r) &= T(n/2, r_1) + T(n/2,r_2) + n\\
&\leq 2\frac{n}{2} \log_2(r_1+1)+2\frac{n}{2} \log_2(r_2+1)+n\log_2 2\\
&= n \log_2 (2(r_1+1)(r_2+1))\\
&\leq n \log_2 (r+1)^2\\
&= 2n \log_2 (r+1).
\end{align*}
The last inequality follows from simple algebra using $r=r_1+r_2$.
Thus, at most $2n(d+1)\log_2 \left(r+1\right)$ operations are needed to apply $\bS^T\bH\bD$ on $d+1$ vectors. After this preprocessing, the \textsc{RandLeastSquares} algorithm must compute the pseudoinverse of an $r \times d$ matrix, or, equivalently, solve a least-squares problem on $r$ constraints and $d$ variables. This operation can be performed in $O(rd^2)$ time since $r \geq d$. Thus, the entire algorithm runs in time
\[n(d+1) + 2n(d+1) \log_2 \left(r + 1\right) +\mathcal{O}\left(rd^2 \right).\]

\subsection{References.} Our presentation in this chapter follows closely the derivations in~\cite{DMMS11}; see~\cite{DMMS11} for a detailed discussion of prior work on this topic. We also refer the interested reader to~\cite{AMT10,Woodruff2014} for followup work on randomized solvers for least-squares problems.

\section{A RandNLA Algorithm for Low-rank Matrix Approximation}\label{sxn:main:lowrank}

In this section, we will present a simple randomized matrix algorithm for low-rank matrix approximation.
Algorithms to compute low-rank approximations to matrices have been of paramount importance historically in scientific computing (see, for example,~\cite{Saad2011} for traditional numerical methods based on subspace iteration and Krylov subspaces to compute such approximations) as well as more recently in machine learning and data analysis. RandNLA has pioneered an alternative approach, by applying random sampling and random projection algorithms to construct such low-rank approximations with provable accuracy guarantees; see~\cite{dkm_matrix2} for early work on the topic and~\cite{Mah-mat-rev_BOOK,HMT09_SIREV,Woodruff2014,MD2016} for overviews of more recent approaches. In this section, we will present and analyze a simple algorithm to approximate the top $k$ left singular vectors of a matrix $\bA \in \mathbb{R}^{m \times n}$.
Many RandNLA methods for low-rank approximation boil down to variants of this basic technique; see, e.g., the chapter by Martinsson in this volume~\cite{pcmi-chapter-martinsson}.
Unlike the previous section on RandNLA algorithms for regression problems, no particular assumptions will be imposed on $m$ and $n$; indeed, $\bA$ could be a square matrix.

\subsection{The main algorithm and main theorem.} \label{sxn:sampling:result2}
Our main algorithm is quite simple and again leverages the Randomized Hadamard Tranform of Section~\ref{sxn:RHT}. Indeed, let the matrix product $\bH \bD$ denote the $n \times n$ Randomized Hadamard Transform.
First, we \textit{postmultiply} the input matrix $\bA \in \mathbb{R}^{m \times n}$ by  $\left(\bH \bD\right)^T$, thus forming a new matrix $\bA \bD \bH \in \mathbb{R}^{m \times n}$.%
\footnote{Alternatively, we could \emph{premultiply} $\bA^T$ by $\bH\bD$.  The reader should become comfortable going back and forth with such manipulations.}
Then, we sample (uniformly at random) $c$ columns from the matrix $\bA \bD \bH $, thus forming a \textit{smaller} matrix $\bC \in \mathbb{R}^{m \times c}$.
Finally, we use a Ritz-Rayleigh type procedure to construct approximations $\tilde{\bU}_k \in \mathbb{R}^{m \times k}$ to the top $k$ left singular vectors of $\bA$ from $\bC$; these approximations lie within the column space of $\bC$.
See the \textsc{RandLowRank} algorithm (Algorithm~\ref{alg:alg_randlowrank}) for a detailed description of this procedure, using a sampling-and-rescaling matrix $\bS \in \mathbb{R}^{n \times c}$ to form the matrix $\bC$. Theorem~\ref{thm:relerrLowRank} is our main quality-of-approximation result for the~\textsc{RandLowRank} algorithm.

\input{alg_RandLowRank}

\begin{theorem}
\label{thm:relerrLowRank}
Let $\bA \in \mathbb{R}^{m \times n}$, let $k$ be a rank parameter, and let $\epsilon \in (0,1/2]$.
If we set
\begin{equation}\label{eqn:cval5}
c \geq c_0\frac{k\ln n}{\epsilon^2} \left(\ln\frac{k}{\epsilon^2}+\ln\ln n\right),
\end{equation}
(for a fixed constant $c_0$) then, with probability at least .85, the \textsc{RandLowRank} algorithm returns a matrix $\tilde{\bU}_k \in \mathbb{R}^{m \times k}$ such that
\begin{equation}
\label{eqn:thm_relerrLowRank}
\FNorm{\bA-\tilde{\bU}_k\tilde{\bU}_k^T\bA} \le (1+\epsilon) \FNorm{\bA-\bU_k\bU_k^T\bA} = (1+\epsilon) \FNorm{\bA-\bA_k}.
\end{equation}
(Here, $\bU_k \in \mathbb{R}^{m \times k}$ contains the top $k$ left singular vectors of $\bA$).
The running time of the \textsc{RandLowRank} algorithm is
$O(mnc)$.
\end{theorem}
We discuss the dimensions of the matrices in steps 6-9 of the \textsc{RandLowRank} algorithm. One can think of the matrix $\bC \in \mathbb{R}^{m \times c}$ as a ``sketch'' of the input matrix $\bA$. Notice that $c$ is (up to $\ln\ln$ factors and ignoring constant terms like $\epsilon$ and $\delta$) $O(k \ln k)$; the rank of $\bC$ (denoted by $\rho_C$) is at least $k$, i.e., $\rho_C \geq k$. The matrix $\bU_C$ has dimensions $m \times \rho_C$ and the matrix $\bW$ has dimensions $\rho_C \times n$. Finally, the matrix $\bU_{W,k}$ has dimensions $\rho_C \times k$ (by our assumption on the rank of $\bW$).

Recall that the \textit{best} rank-$k$ approximation to $\bA$ is equal to $\bA_k=\bU_k\bU_k^T\bA$. In words, Theorem~\ref{thm:relerrLowRank} argues that the \textsc{RandLowRank} algorithm returns a set of $k$ orthonormal vectors that are excellent approximations to the top $k$ left singular vectors of $\bA$, in the sense that projecting $\bA$ on the subspace spanned by $\tilde{\bU}_k$ returns a matrix that has residual error that is close to that of $\bA_k$.

\begin{Remark}
We emphasize that the $O(mnc)$ running time of the \textsc{RandLowRank} algorithm is due to the Ritz-Rayleigh type procedure in steps (7)-(9). These steps guarantee that the proposed algorithm returns a matrix $\tilde{\bU}_k$ with \textit{exactly} $k$ columns that approximates the top $k$ left singular vectors of $\bA$. The results of~\cite{pcmi-chapter-martinsson} focus (in our parlance) on the matrix $\bC$, which can be constructed much faster (see Section~\ref{sxn:ch4:runningtime}), in $O(mn\log_2 c)$ time, but has \textit{more than} $k$ columns. One can bound the error term $\FNorm{\bA-\bC\bC^{\dagger}\bA}=\FNorm{\bA-\bU_C\bU_C^T\bA}$ to prove that the column span of $\bC$ contains good approximations to the top $k$ left singular vectors of $\bA$.
\end{Remark}

\begin{Remark}
Repeating the \textsc{RandLowRank} algorithm $\lceil\ln(1/\delta)/\ln5\rceil$ times and keeping the matrix $\tilde{\bU}_k$ that minimizes the error
$\FNorm{\bA-\tilde{\bU}_k\tilde{\bU}_k^T\bA}$ reduces the failure probability of the algorithm to at most $1-\delta$, for any $\delta\in(0,1)$.
\end{Remark}

\begin{Remark}
As with the sampling process in the \textsc{RandLeastSquares} algorithm, the operation represented by $\bD \bH \bS$ in the \textsc{RandLowRank} algorithm can be viewed in one of two equivalent ways: either as a random preconditioning followed by a uniform sampling operation; or as a Johnson-Lindenstrauss style random projection.
(In particular, informally, the \textsc{RandLowRank} algorithm ``works'' for the following reason.
If a matrix is well-approximated by a low-rank matrix, then there is redundancy in the columns (and/or rows), and thus random sampling ``should'' be successful at selecting a good set of columns.
That said, just as with the \textsc{RandLeastSquares} algorithm, there may be some columns that are more important to select, e.g., that have high leverage.  Thus, using a random projection, which transforms the input to a new basis where the leverage scores of different columns are uniformized, amounts to preconditioning the input such that uniform sampling is appropriate.)
\end{Remark}

\begin{Remark}
The value $c$ is essentially\footnote{We omit the $\ln \ln n$ term from this qualitative remark. Recall that $\ln \ln n$ goes to infinity with dignity and therefore, quoting Stan Eisenstat, $\ln \ln n$ is for all practical purposes essentially a constant; see~\url{https://rjlipton.wordpress.com/2011/01/19/we-believe-a-lot-but-can-prove-little/}.} equal to $O((k/\epsilon^2)\ln(k/\epsilon)\ln n)$. For constant $\epsilon$, this grows as a function of $k\ln k$ and $\ln n$.
\end{Remark}

\begin{Remark}
Similar bounds can be proven for many other random projection algorithms (using different values for $c$) and not just the Randomized Hadamard Transform.
Well-known alternatives include random Gaussian matrices, the Randomized Discrete Cosine Transform, sparsity-preserving random projections, etc.
Which variant is most appropriate in a given situation depends on the sparsity structure of the matrix, the noise properties of the data, the model of data access, etc.
See~\cite{Mah-mat-rev_BOOK,Woodruff2014} for an overview of similar results.
\end{Remark}

\begin{Remark}
One can generalize the \textsc{RandLowRank} algorithm to work with the matrix $(\bA\bA^T)^t\bA\bD \bH \bS$ for integer $t\geq 0$. This would result in subspace iteration. If all intermediate iterates (for $t=0,1,\ldots$) are kept, the Krylov subspace would be formed. See~\cite{MM2015,DrineasIKM16} and references therein for a detailed treatment and analysis of such methods.
(See also the chapter by Martinsson in this volume~\cite{pcmi-chapter-martinsson} for related results.)
\end{Remark}

The remainder of this section will focus on the proof of Theorem~\ref{thm:relerrLowRank}.
Our proof strategy will consist of three steps.
First (Section~\ref{s_aux1}), we we will prove that:
\[\FNormS{\bA-\tilde\bU_k\tilde\bU_k^T\bA} \leq
 \FNormS{\bA_k-\bU_C \bU_C^T\bA_k} + \FNormS{\bA_{k,\perp}}.\]
The above inequality allows us to manipulate the easier-to-bound term $\FNormS{\bA_k-\bU_C \bU_C^T\bA_k}$ instead of the term $\FNormS{\bA-\tilde\bU_k\tilde\bU_k^T\bA}$.
Second (Section~\ref{s_aux2}), to bound this term, we will use a structural inequality that is central (in this form or mild variations) in many RandNLA low-rank approximation algorithms and their analyses. Indeed, we will argue that
\begin{align*}
\FNormS{\bA_k-\bU_C \bU_C^T\bA_k}
&\leq2\FNormS{(\bA-\bA_k)\bD\bH\bS((\bV_k^T\bD\bH\bS)^{\dagger}-(\bV_k^T\bD\bH\bS)^T)}\\
&\hspace{15mm}+2\FNormS{(\bA-\bA_k)\bD\bH\bS(\bV_k^T\bD\bH\bS)^T}.
\end{align*}
Third (Section~\ref{s_aux3}), we will use results from Section~\ref{chapter:MM} to bound the two terms at the right hand side of the above inequality.

\subsection{An alternative expression for the error.}\label{s_aux1}
The \textsc{RandLowRank} algorithm approximates the top $k$ left singular vectors of $\bA$, i.e., the matrix $\bU_k \in \mathbb{R}^{m \times k}$,
by the orthonormal matrix $\tilde\bU_k \in \mathbb{R}^{m \times k}$.
Bounding $\|\bA-\tilde\bU_k\tilde\bU_k^T\bA\|_F$ directly seems hard, so we present an alternative expression that is easier to analyze and that also reveals an interesting insight for $\tilde{\bU}_k$. We will prove that the matrix $\tilde\bU_k \tilde\bU_k \bA$ is the best rank-$k$ approximation to $\bA$ (with respect to the Frobenius norm\footnote{This is not true for other unitarily invariant norms, e.g., the two-norm; see~\cite{BDM2014} for a detailed discussion.}) that lies within the column space of the matrix $\bC$. This optimality property is guaranteed by the Ritz-Rayleigh type procedure implemented in Steps 7-9 of the~\textsc{RandLowRank} algorithm.

\begin{lemma}\label{lem:restate}
Let $\bU_C$ be a basis for the column span of $\bC$ and let $\tilde{\bU}_k$ be the output of the~\textsc{RandLowRank} algorithm. Then
\begin{equation}
\bA-\tilde\bU_k\tilde\bU_k^T\bA = \bA-\bU_C \left(\bU_C^T\bA\right)_k.
\end{equation}
In addition, $\bU_C (\bU_C^T\bA)_k$ is the best rank-$k$ approximation to $\bA$, with respect to the Frobenius norm, that lies within the column span of the matrix $\bC$, namely
\begin{equation}\label{eqn:opt1}
\FNormS{\bA-\bU_C (\bU_C^T\bA)_k} =
\min_{\rank{\bY}\leq k}\FNormS{\bA - \bU_C \bY}.
\end{equation}
\end{lemma}
\begin{proof}
Recall that $\tilde{\bU}_k=\bU_C\bU_{W,k}$, where
$\bU_{W,k}$ is the matrix of the top $k$ left singular vectors of
$\bW = \bU_C^T\bA$. Thus, $\bU_{W,k}$ spans the same range as $\bW_k$,
the best rank-$k$ approximation to $\bW$, i.e., $\bU_{W,k}\bU_{W,k}^T=\bW_k \bW_k^{\dagger}$. Therefore
\begin{align*}
\bA-\tilde\bU_k\tilde\bU_k^T\bA &= \bA-\bU_C \bU_{W,k}\bU_{W,k}^T\bU_C^T\bA\\
 &= \bA-\bU_C \bW_k \bW_k^{\dagger} \bW = \bA-\bU_C \bW_k.
\end{align*}
The last equality follows from $\bW_k\bW_k^{\dagger}$ being the orthogonal projector onto the range of $\bW_k$.
In order to prove the optimality property of the lemma, we simply observe that
\begin{align*}
\FNormS{\bA-\bU_C (\bU_C^T\bA)_k} &= \FNormS{\bA-\bU_C\bU_C^T\bA + \bU_C\bU_C^T\bA - \bU_C (\bU_C^T\bA)_k}\\
&=\FNormS{(\bI-\bU_C\bU_C^T)\bA + \bU_C(\bU_C^T\bA - (\bU_C^T\bA)_k)}\\
&=\FNormS{(\bI-\bU_C\bU_C^T)\bA} + \FNormS{\bU_C(\bU_C^T\bA - (\bU_C^T\bA)_k)}\\
&=\FNormS{(\bI-\bU_C\bU_C^T)\bA} + \FNormS{\bU_C^T\bA - (\bU_C^T\bA)_k}.
\end{align*}
The second to last equality follows from Matrix Pythagoras (Lemma~\ref{l_pyth}) and the last equality follows from the orthonormality of the columns of $\bU_C$. The second statement of the lemma is now immediate since $(\bU_C^T\bA)_k$ is the best rank-$k$ approximation to $\bU_C^T\bA$ and thus any other matrix $\bY$ of rank at most $k$ would result in a larger Frobenius norm error.
\end{proof}

Lemma~\ref{lem:restate} shows that Eqn.~(\ref{eqn:thm_relerrLowRank}) in Theorem~\ref{thm:relerrLowRank}
can be proven by bounding $\|\bA-\bU_C \left(\bU_C^T\bA\right)_k\|_F$.
Next, we transition from the best rank-$k$ approximation of the projected matrix $(\bU_C^T\bA)_k$
to the best rank-$k$ approximation $\bA_k$ of the original matrix. First (recall the notation introduced in Section~\ref{sxn:review:SVD}), we split
\begin{align}\label{e_Ai}
\bA=\bA_k+ \bA_{k,\perp},\ \text{where}\ \bA_k=\bU_k\bSigma_k\bV_k^T \  \text{and}\
\bA_{k,\perp}=\bU_{k,\perp}\bSigma_{k,\perp}\bV_{k,\perp}^T.
\end{align}

\begin{lemma}\label{l_aux2}
Let $\bU_C$ be an orthonormal basis for the column span of the matrix $\bC$ and let $\tilde{\bU}_k$ be the output of the \textsc{RandLowRank} algorithm. Then,
\[\FNormS{\bA-\tilde\bU_k\tilde\bU_k^T\bA} \leq
 \FNormS{\bA_k-\bU_C \bU_C^T\bA_k} + \FNormS{\bA_{k,\perp}}.\]
\end{lemma}
\begin{proof}
The optimality property in Eqn.~(\ref{eqn:opt1}) in Lemma~\ref{lem:restate} and the fact that $\bU_C^T\bA_k$ has rank at most $k$ imply
\begin{align*}
\FNormS{\bA-\tilde\bU_k\tilde\bU_k^T\bA} &=
\FNormS{\bA-\bU_C \left(\bU_C^T\bA\right)_k}\nonumber\\
&\leq \FNormS{\bA-\bU_C \bU_C^T\bA_k} \nonumber\\
&= \FNormS{\bA_k-\bU_C\bU_C^T\bA_k} + \FNormS{\bA_{k,\perp}}.
\end{align*}
The last equality follows from Lemma~\ref{l_pyth}.
\end{proof}

\subsection{A structural inequality.}\label{s_aux2}
We now state and prove a structural inequality that will help us bound $\FNormS{\bA_k-\bU_C \bU_C^T\bA_k}$ (the first term in the error bound of Lemma~\ref{l_aux2}).
This structural inequality, or minor variants of it, underlie nearly all RandNLA algorithms for low-rank matrix approximation~\cite{MD2016}.
To understand this structural inequality, recall that, given a matrix $\bA \in \mathbb{R}^{m \times n}$, many RandNLA algorithms seek to construct a ``sketch'' of $\bA$ by post-multiplying $\bA$ by some ``sketching'' matrix $\bZ \in \mathbb{R}^{n \times c}$, where $c$ is much
smaller than $n$.
(In particular, this is precisely what the \textsc{RandLowRank} algorithm does.)
Thus, the resulting matrix $\bA\bZ \in \mathbb{R}^{m \times c}$ is much smaller than the original matrix $\bA$, and the interesting question
is the approximation guarantees that it offers.

A common approach is to explore how well $\bA\bZ$ spans the principal subspace
of $\bA$, and one metric of accuracy is some norm of the error matrix $\bA_k - (\bA\bZ)(\bA\bZ)^{\dagger}\bA_k$,
where $(\bA\bZ)(\bA\bZ)^{\dagger} \bA_k$ is the projection of $\bA_k$ onto the subspace spanned by the
columns of $\bA\bZ$. (See Section~\ref{sxn:review:MP} for the definition of the Moore-Penrose
pseudoinverse of a matrix.)
The following structural result offers a means to bound the Frobenius norm of the error matrix $\bA_k - (\bA\bZ)(\bA\bZ)^{\dagger}\bA_k$.
\begin{lemma}
\label{lem:mainlemma_general}
Given $\bA \in \mathbb{R}^{m \times n}$, let $\bZ \in \mathbb{R}^{n \times c}$ ($c \geq k$) be any matrix such that
$\bV_k^T\bZ \in \mathbb{R}^{k \times c}$ has rank $k$. Then,
\begin{equation}
\FNormS{\bA_k - (\bA\bZ)(\bA\bZ)^{\dagger}\bA_k}
   \leq \FNormS{\left(\bA-\bA_k\right)\bZ(\bV_k^T\bZ)^{\dagger}}.
\label{eqn:struct-cond-low-rank-gen}
\end{equation}
\end{lemma}

\begin{Remark}
Lemma~\ref{lem:mainlemma_general} holds for \emph{any} matrix $\bZ$,
regardless of whether $\bZ$ is constructed deterministically or randomly.
In the context of RandNLA, typical constructions of $\bZ$ would represent a
random sampling or random projection operation, like the the matrix $\bD\bH\bS$ used in the \textsc{RandLowRank} algorithm.
\end{Remark}

\begin{Remark}
The lemma actually holds for any unitarily invariant norm, including the two and the nuclear norm of a matrix~\cite{MD2016}.
\end{Remark}

\begin{Remark}
See~\cite{MD2016} for a detailed discussion of such structural inequalities and their history. Lemma~\ref{lem:mainlemma_general} immediately suggests a proof strategy for bounding the error of RandNLA algorithms for low-rank matrix approximation: identify a sketching matrix $\bZ$ such that $\bV_k^T\bZ$ has full rank; and, at the same time, bound the relevant norms of $\left(\bV_k^T\bZ\right)^\dagger$ and $(\bA-\bA_k) \bZ$.
\end{Remark}

\begin{proof}(of Lemma~\ref{lem:mainlemma_general})
First, note that
\[(\bA\bZ)^{\dagger}\bA_k = \mbox{argmin}_{\bX \in \mathbb{R}^{c \times n}} \FNormS{\bA_k - \left(\bA\bZ\right)\bX }.\]
The above equation follows by viewing the above optimization problem as least-squares regression with multiple right-hand sides. Interestingly, this property holds for any unitarily invariant norm, but the proof is involved; see Lemma 4.2 of~\cite{DrineasIKM16} for a detailed discussion.
This implies that instead of bounding $\FNormS{\bA_k - (\bA\bZ)(\bA\bZ)^{\dagger}\bA_k}$, we can replace $\left(\bA\bZ\right)^+\bA_k$ with any other $c \times n$ matrix and the equality with an inequality. In particular, we replace $\left(\bA\bZ\right)^{\dagger} \bA_k$ with $\left(\bA_k\bZ\right)^{\dagger} \bA_k$:
\begin{align}
\nonumber \FNormS{\bA_k - (\bA\bZ)(\bA\bZ)^{\dagger}\bA_k}
\nonumber   &\leq  \FNormS{ \bA_k - \bA\bZ \left(\bA_k \bZ\right)^{\dagger} \bA_k}.
\end{align}
This suboptimal choice for $\bX$ is essentially the ``heart'' of our proof: it allows us to manipulate and further decompose the error term, thus making the remainder of the analysis feasible.
Use $\bA = \bA-\bA_k+\bA_k$ to get
\begin{align}
\nonumber
\FNormS{\bA_k - (\bA\bZ)(\bA\bZ)^{\dagger}\bA_k}
   &\leq \FNormS{ \bA_k - (\bA- \bA_k+ \bA_k) \bZ \left(\bA_k\bZ\right)^{\dagger} \bA_k} \\
\nonumber
   &= \FNormS{\bA_k-\bA_k\bZ (\bA_k\bZ)^{\dagger} \bA_k-(\bA-\bA_k)\bZ(\bA_k\bZ)^{\dagger} \bA_k}\\
\nonumber
   &= \FNormS{(\bA-\bA_k)\bZ(\bA_k\bZ)^{\dagger} \bA_k}.
\end{align}
To derive the last inequality, we used
\begin{align}
\nonumber \bA_k - \bA_k\bZ (\bA_k\bZ)^{\dagger} \bA_k
   &= \bA_k - \bU_k\bSigma_k\bV_k^T\bZ (\bU_k\bSigma_k\bV_k^T\bZ)^{\dagger} \bA_k \\
\label{eq31}
   &= \bA_k - \bU_k\bSigma_k(\bV_k^T\bZ) (\bV_k^T\bZ)^{\dagger}\bSigma_k^{-1}\bU_k^T \bA_k \\
\label{eq32}
   &= \bA_k - \bU_k\bU_k^T \bA_k=\bzero.
\end{align}
In Eqn.~(\ref{eq31}), we used Eqn.~(\ref{eqn:pinv}) and the fact that both matrices $\bV_k^T\bZ$ and $\bU_k\bSigma_k$ have rank $k$. The latter fact also implies that $(\bV_k^T\bZ) (\bV_k^T\bZ)^{\dagger}=\bI_k$, which derives Eqn.~(\ref{eq32}). Finally, the fact that $\bU_k\bU_k^T \bA_k=\bA_k$ concludes the derivation and the proof of the lemma.
\end{proof}

\subsection{Completing the proof of Theorem~\ref{thm:relerrLowRank}.}
\label{s_aux3}
In order to complete the proof of the relative error guarantee of Theorem~\ref{thm:relerrLowRank}, we will complete the strategy outlined at the end of Section~\ref{sxn:sampling:result2}.
First, recall that from Lemma~\ref{l_aux2} it suffices to bound
\begin{equation}\label{eqn:drin1}
\FNormS{\bA-\tilde\bU_k\tilde\bU_k^T\bA} \leq \FNormS{\bA_k-\bU_C \bU_C^T\bA_k} + \FNormS{\bA-\bA_{k}}.
\end{equation}
Then, to bound the first term in the right-hand side of the above inequality, we will apply the structural result of Lemma~\ref{lem:mainlemma_general} on the matrix $$ \bPhi = \bA\bD\bH , $$ with $\bZ=\bS$, where the matrices $\bD$, $\bH$, and $\bS$ are constructed as
described in the \textsc{RandLowRank} algorithm.
Lemma~\ref{lem:mainlemma_general} states that if $\bV_{\Phi,k}^T\bS$ has rank $k$, then
\begin{equation}
\FNormS{\bPhi_k - (\bPhi\bS)(\bPhi\bS)^{\dagger}\bPhi_k}\leq
 \FNormS{(\bPhi-\bPhi_k)\bS(\bV_{\Phi,k}^T\bS)^{\dagger}}.
\label{eqn:struct-cond-low-rank-gen:phi}
\end{equation}
Here, we used $\bV_{\Phi,k} \in \mathbb{R}^{n \times k}$ to denote the matrix of the top $k$ right singular vectors of $\bPhi$. Recall from Section~\ref{sxn:RHT} that $\bD\bH$ is an orthogonal matrix and thus the left singular vectors and the singular values of the matrices $\bA$ and $\bPhi=\bA\bD\bH$ are identical. The right singular vectors of the matrix $\bPhi$ are simply the right singular vectors of $\bA$, rotated by $\bD\bH$, namely
\[\bV_{\bPhi}^T=\bV^T\bD\bH,\]
where $\bV$ (respectively, $\bV_{\bPhi}$) denotes the matrix of the right singular vectors of $\bA$ (respectively, $\bPhi$). Thus, $\bPhi_k=\bA_k\bD\bH$, $\bPhi-\bPhi_k=(\bA-\bA_k)\bD\bH$, and $\bV_{\Phi,k}=\bV_{k}\bD\bH$. Using all the above, we can rewrite Eqn.~(\ref{eqn:struct-cond-low-rank-gen:phi}) as follows:
\begin{equation}
\FNormS{\bA_k - (\bA\bD\bH\bS)(\bA\bD\bH\bS)^{\dagger}\bA_k}\leq
 \FNormS{(\bA-\bA_k)\bD\bH\bS(\bV_{k}^T\bD\bH\bS)^{\dagger}}.
\label{eqn:struct-cond-low-rank-gen:phi2}
\end{equation}
In the above derivation, we used unitary invariance to drop a $\bD\bH$ term from the Frobenius norm. Recall that $\bA_{k,\perp}=\bA-\bA_k$;
we now proceed to manipulate the right-hand side of the above inequality as follows\footnote{We use the following easy-to-prove version of the triangle inequality for the Frobenius norm: for any two matrices $\bX$ and $\bY$ that have the same dimensions, $\FNormS{\bX+\bY}\leq 2\FNormS{\bX}+2\FNormS{\bY}$.}:
\begin{align}
\nonumber\FNormS{\bA_{k,\perp}\bD\bH\bS(\bV_k^T\bD\bH\bS)^{\dagger}}
 \\
\nonumber
 &\hspace{-20mm}=\FNormS{\bA_{k,\perp}\bD\bH\bS((\bV_k^T\bD\bH\bS)^{\dagger}-(\bV_k^T\bD\bH\bS)^T+(\bV_k^T\bD\bH\bS)^T)}\\
\label{eqn:ch4ppp5}&\hspace{-20mm}\leq 2\FNormS{\bA_{k,\perp}\bD\bH\bS((\bV_k^T\bD\bH\bS)^{\dagger}-(\bV_k^T\bD\bH\bS)^T)}\\
\label{eqn:ch4ppp6}&\hspace{-10mm}+ 2\FNormS{\bA_{k,\perp}\bD\bH\bS(\bV_k^T\bD\bH\bS)^T}.
\end{align}
We now proceed to bound the terms in~(\ref{eqn:ch4ppp5}) and~(\ref{eqn:ch4ppp6}) separately.
Our first order of business, however, will be to quantify the manner in which the Randomized Hadamard Transform approximately uniformizes information in the top $k$ right singular vectors of $\bA$.

\subsubsection*{The effect of the Randomized Hadamard Transform.}
Here, we state a lemma that quantifies the manner in which $\bH \bD$ (premultiplying $\bV_k$, or $\bD \bH$ postmultiplying $\bV_k^T$) approximately ``uniformizes'' information in the right singular subspace of the matrix $\bA$, thus allowing us to apply our matrix multiplication results from Section~\ref{chapter:MM} in order to bound~(\ref{eqn:ch4ppp5}) and~(\ref{eqn:ch4ppp6}).
This is completely analogous to our discussion in Section~\ref{sxn:ls:thmproof} regarding the \textsc{RandLeastSquares} algorithm.
\begin{lemma} \label{lem:HUnew}
Let $\bV_k$ be an $n \times k$ matrix with orthonormal columns and let the product $\bH \bD$ be the $n \times n$ Randomized Hadamard Transform of Section~\ref{sxn:RHT}. Then, with probability at least $.95$,
\begin{align}
\label{eqn:lem:HU_eqn2new} \TNormS{\left(\bH \bD \bV_k \right)_{i*}}
   &\leq \frac{2k\ln(40nk)}{n},\qquad \text{ for all } i =1,\ldots, n   .
\end{align}
\end{lemma}
The proof of the above lemma is identical to the proof of Lemma~\ref{lem:HU}, with $\bV_k$ instead of $\bU$ and $k$ instead of $d$.

\subsubsection{Bounding Expression~(\ref{eqn:ch4ppp5}).}
To bound the term in Expression~(\ref{eqn:ch4ppp5}), we first use the strong submultiplicativity of the Frobenius norm (see Section~\ref{sxn:review:FNorm}) to~get
\begin{align}
\nonumber
&
\hspace{-20mm}
\FNormS{\bA_{k,\perp}\bD\bH\bS((\bV_k^T\bD\bH\bS)^{\dagger}-(\bV_k^T\bD\bH\bS)^T)}  \\
\leq
\label{eqn:n1ppp5}&\FNormS{\bA_{k,\perp}\bD\bH\bS}\TNormS{(\bV_k^T\bD\bH\bS)^{\dagger}-(\bV_k^T\bD\bH\bS)^T}.
\end{align}
Our first lemma bounds the term $\FNormS{(\bA-\bA_k)\bD\bH\bS} = \FNormS{\bA_{k,\perp}\bD\bH\bS}$. We actually prove the result for any matrix $\bX$ and for our choice for the matrix $\bS$ in the \textsc{RandLowRank} algorithm.
\begin{lemma}\label{lem:FNormExp}
Let the sampling matrix $\bS \in \mathbb{R}^{n \times c}$ be constructed as in the \textsc{RandLowRank} algorithm.
Then, for any matrix $\bX\in \mathbb{R}^{m\times n}$,
\[
\Expect{\FNormS{\bX\bS}}=\FNormS{\bX}  ,
\] and, from Markov's inequality (see Section~\ref{sxn:review:Markov}),
with probability at least 0.95,
\[\FNormS{\bX\bS} \leq 20\FNormS{\bX}.\]
\end{lemma}

\begin{Remark} The above lemma holds even if the sampling of the canonical vectors $\be_i$ to be included in $\bS$ is not done uniformly at random, but with respect to any set of probabilities $\{p_1,\ldots,p_n\}$ summing up to one, as long as the selected canonical vector at the $t$-th trial (say the $i_t$-th canonical vector $\be_{i_t}$) is rescaled by $\sqrt{1/cp_{i_t}}$.
Thus, even for nonuniform sampling, $\bX\bS$ is an unbiased estimator for the Frobenius norm of the matrix $\bX$.
\end{Remark}

\begin{proof}(of Lemma~\ref{lem:FNormExp})
We compute the expectation of $\FNormS{\bX\bS}$ from first principles as follows:
\begin{align*}
\Expect{\FNormS{\bX\bS}}&=&\sum_{t=1}^c\sum_{j=1}^n \frac{1}{n}\cdot\FNormS{\sqrt{\frac{n}{c}}\bX_{*j}}=
\frac{1}{c}\sum_{t=1}^c\sum_{j=1}^n \FNormS{\bX_{*j}}=\FNormS{\bX}.
\end{align*}
The lemma now follows by applying Markov's inequality.
\end{proof}

We can now prove the following lemma, which will be conditioned on Eqn.~(\ref{eqn:lem:HU_eqn2new}) holding.
\begin{lemma}
\label{lem:sample_lem40pfn1}
Assume that Eqn.~(\ref{eqn:lem:HU_eqn2new}) holds.  If $c$ satisfies
\begin{equation}\label{eqn:ppdd2}
c \geq \frac{192k\ln(40nk)}{\epsilon^2} \ln\left(\frac{192\sqrt{20}k\ln(40nk)}{\epsilon^2}\right),
\end{equation}
then with probability at least .95,
\[\TNormS{\left(\bV_k^T\bD\bH\bS\right)^{\dagger}-\left(\bV_k^T\bD\bH\bS\right)^T}\leq2\epsilon^2.\]
\end{lemma}

\begin{proof}
Let $\sigma_i$ denote the $i$-th singular value of the matrix $\bV_k^T\bD\bH\bS$. Conditioned on Eqn.~(\ref{eqn:lem:HU_eqn2new}) holding, we can replicate the proof of Lemma~\ref{lem:sample_lem20pf} to argue that if $c$ satisfies Eqn.~(\ref{eqn:ppdd2}), then, with probability at least .95,
\begin{equation}\label{eqn:ppdd1}
\abs{1 - \sigma_i^2} \leq \epsilon
\end{equation}
holds for all $i$. (Indeed, we can replicate the proof of Lemma~\ref{lem:sample_lem20pf} using $\bV_k$ instead of $\bU_A$ and $k$ instead of $d$; we also evaluate the bound for arbitrary $\epsilon$ instead of fixing it.)
We now observe that the matrices
$$
\left(\bV_k^T\bD\bH\bS\right)^{\dagger} \quad\mbox{and}\quad \left(\bV_k^T\bD\bH\bS\right)^T
$$
have the same left and right singular vectors\footnote{Given any matrix $\bX$ with thin SVD $\bX=\bU_X\bSigma_X\bV_X^T$ its transpose is $\bX^T=\bV_X\bSigma_X\bU_X^T$ and its pseudoinverse is $\bX=\bV_X\bSigma_X^{-1}\bU_X^T$.}. Recall that in this lemma we used $\sigma_i$ to denote the singular values of the matrix $\bV_k^T\bD\bH\bS$.
Then, the singular values of the matrix $\left(\bV_k^T\bD\bH\bS\right)^T$ are equal to the $\sigma_i$'s, while the singular values of the matrix $\left(\bV_k^T\bD\bH\bS\right)^{\dagger}$ are equal to $\sigma_i^{-1}$. Thus,
\[\TNormS{\left(\bV_k^T\bD\bH\bS\right)^{\dagger}-\left(\bV_k^T\bD\bH\bS\right)^T} = \max_i\abs{\sigma_i^{-1}-\sigma_i}^2=
\max_i\abs{(1-\sigma_i^2)^2\sigma_i^{-2}}.\]
Combining with Eqn.~(\ref{eqn:ppdd1}) and using the fact that $\epsilon \leq 1/2$,
$$\TNormS{\left(\bV_k^T\bD\bH\bS\right)^{\dagger}-\left(\bV_k^T\bD\bH\bS\right)^T}=
\max_i\abs{(1-\sigma_i^2)\sigma_i^{-2}}\leq (1-\epsilon)^{-1}\epsilon^2\leq2\epsilon^2.$$
\end{proof}

\begin{lemma}\label{lem:sample_lem40pfn2}
Assume that Eqn.~(\ref{eqn:lem:HU_eqn2new}) holds.  If $c$ satisfies Eqn.~(\ref{eqn:ppdd2}), then, with probability at least .9,
\[2\FNormS{\bA_{k,\perp}\bD\bH\bS((\bV_k^T\bD\bH\bS)^{\dagger}-(\bV_k^T\bD\bH\bS)^T)}\leq 80\epsilon^2\FNormS{\bA_{k,\perp}}.\]
\end{lemma}

\begin{proof}
We combine Eqn.~(\ref{eqn:n1ppp5}) with Lemmas~\ref{lem:FNormExp} (applied to $\bX=\bA_{k,\perp}$) and Lemma~\ref{lem:sample_lem40pfn1} to get that, conditioned on Eqn.~(\ref{eqn:lem:HU_eqn2new}) holding, with probability at least 1-0.05-0.05=0.9,
\[\FNormS{\bA_{k,\perp}\bD\bH\bS((\bV_k^T\bD\bH\bS)^{\dagger}-(\bV_k^T\bD\bH\bS)^T)}\leq 40\epsilon^2\FNormS{\bA_{k,\perp}}.\]
The aforementioned failure probability follows from a simple union bound on the failure probabilities of
Lemmas~\ref{lem:FNormExp} and~\ref{lem:sample_lem40pfn1}.
\end{proof}

\subsubsection*{Bounding Expression~(\ref{eqn:ch4ppp6}).}
Our bound for Expression~(\ref{eqn:ch4ppp6}) will be conditioned on Eqn.~(\ref{eqn:lem:HU_eqn2new}) holding; then, we will use our matrix multiplication results from Section~\ref{chapter:MM} to derive our bounds. Our discussion is completely analogous to the proof of Lemma~\ref{lem:sample_lem40pf}. We will prove the following lemma.

\begin{lemma}
\label{lem:sample_lem40pfn3}
Assume that Eqn.~(\ref{eqn:lem:HU_eqn2new}) holds.  If $c \geq 40k\ln(40nk)/\epsilon$, then, with probability at least .95,
\[\FNormS{\bA_{k,\perp}\bD\bH\bS(\bV_k^T\bD\bH\bS)^T} \leq
\epsilon\FNormS{\bA_{k,\perp}}.\]
\end{lemma}
\begin{proof}
To prove the lemma, we first observe that
\[\FNormS{\bA_{k,\perp}\bD\bH\bS(\bV_k^T\bD\bH\bS)^T}=\FNormS{\bA_{k,\perp}\bD\bH\bS\bS^T\bH^T\bD\bV_k-\bA_{k,\perp}\bD\bH\bH^T\bD\bV_k},\]
since $\bD\bH\bH^T\bD=\bI_n$ and $\bA_{k,\perp}\bV_k=\bzero$.
Thus, we can view $\bA_{k,\perp}\bD\bH\bS\bS^T\bH^T\bD\bV_k$ as approximating the product of two matrices, $\bA_{k,\perp}\bD\bH$ and $\bH^T\bD\bV_k$, by randomly sampling columns from the first matrix and the corresponding rows from the second matrix. Note that the sampling probabilities are uniform and do not depend on the two matrices involved in the product. We will apply the bounds of Eqn.~(\ref{eqn:appopt3result}), after arguing that the assumptions of Eqn.~(\ref{eqn:appopt3}) are satisfied. Indeed, since we condition on Eqn.~(\ref{eqn:lem:HU_eqn2new}) holding, the rows of $\bH^T\bD\bV_k=\bH\bD\bV_k$ satisfy
\begin{equation}
\frac{1}{n}   \ge  \beta \frac{\TNormS{\left(\bH\bD\bV_k\right)_{i*}}}{k}, \qquad       \text{ for all } i =1,\ldots, n,
\end{equation}
for $\beta = \left(2\ln(40nk)\right)^{-1}$. Thus, Eqn.~(\ref{eqn:appopt3result}) implies
\begin{align*}
\Expect{\FNormS{\bA_{k,\perp}\bD\bH\bS\bS^T\bH^T\bD\bV_k-\bA_{k,\perp}\bD\bH\bH^T\bD\bV_k}}
   &\leq \frac{1}{\beta c}\FNormS{\bA_{k,\perp}\bD\bH}\FNormS{\bH \bD \bV_k}\\
   &=     \frac{k}{\beta c}\FNormS{\bA_{k,\perp}}.
\end{align*}
In the above we used $\FNormS{\bH \bD \bV_k} = k$. Markov's inequality now implies that with probability at least .95,
\begin{equation*}
\FNormS{\bA_{k,\perp}\bD\bH\bS\bS^T\bH^T\bD\bV_k-\bA_{k,\perp}\bD\bH\bH^T\bD\bV_k} \leq
\frac{20k}{\beta c}\FNormS{\bA_{k,\perp}}.
\end{equation*}
Setting $r \geq 20k/(\beta\epsilon)$ and using the value of $\beta$ specified above concludes the proof of the lemma.
\end{proof}

\subsubsection*{Concluding the proof of Theorem~\ref{thm:relerrLowRank}.}
We are now ready to conclude the proof, and therefore we revert back to using $\bA_{k,\perp}=\bA-\bA_k$. We first state the following lemma.
\begin{lemma}
\label{lem:sample_lem40pfn4}
Assume that Eqn.~(\ref{eqn:lem:HU_eqn2new}) holds. There exists a constant $c_0$ such that, if
\begin{equation}\label{eqn:cvalfinal}
c \geq c_0\frac{k\ln n}{\epsilon^2} \ln\left(\frac{k\ln n}{\epsilon^2}\right),
\end{equation}
then with probability at least .85,
\[\FNorm{\bA-\tilde\bU_k\tilde\bU_k^T\bA} \leq (1+\epsilon)\FNorm{\bA-\bA_k}.\]
\end{lemma}

\begin{proof}
Combining Lemma~\ref{l_aux2} with Expressions~(\ref{eqn:ch4ppp5}) and~(\ref{eqn:ch4ppp6}), and Lemmas~\ref{lem:sample_lem40pfn2} and~\ref{lem:sample_lem40pfn3},
we get
\[\FNormS{\bA-\tilde\bU_k\tilde\bU_k^T\bA} \leq
 (1+\epsilon+80\epsilon^2)\FNormS{\bA-\bA_k}\leq (1+41\epsilon)\FNormS{\bA-\bA_k}.\]
The last inequality follows by using $\epsilon\leq 1/2$. Taking square roots of both sides and using $\sqrt{1+41\epsilon}\leq 1+21\epsilon$, we get
\[\FNorm{\bA-\tilde\bU_k\tilde\bU_k^T\bA} \leq
 (1+\epsilon+80\epsilon^2)\FNormS{\bA-\bA_k}\leq (1+21\epsilon)\FNorm{\bA-\bA_k}.\]
Observe that $c$ has to be set to the maximum of the values used in Lemmas~\ref{lem:sample_lem40pfn2} and~\ref{lem:sample_lem40pfn3}, which is the value of Eqn.~(\ref{eqn:ppdd2}). Adjusting $\epsilon$ to $\epsilon/21$ and appropriately adjusting the constants in the expression of $c$ concludes the lemma.
(We made no particular effort to compute or optimize the constant $c_0$ in the expression of $c$.)

The failure probability follows by a union bound on the failure probabilities of Lemmas~\ref{lem:sample_lem40pfn2} and~\ref{lem:sample_lem40pfn3} conditioned on Eqn.~(\ref{eqn:lem:HU_eqn2new}).
\end{proof}

To conclude the proof of Theorem~\ref{thm:relerrLowRank}, we simply need to remove the conditional probability from Lemma~\ref{lem:sample_lem40pfn4}. Towards that end, we follow the same strategy as in Section~\ref{sxn:ls:thmproof}, to conclude that the success probability of the overall approach is at least $0.85\cdot 0.95 \geq 0.8$.

\subsection{Running time.}\label{sxn:ch4:runningtime}
The~\textsc{RandLowRank} algorithm computes the product $\bC=\bA\bH\bD\bS$ using the ideas of Section~\ref{sxn:lsruntime}, thus taking $2n(m+1)\log_2(c+1)$ time. Step 7 takes $O(mc^2)$; step 8 takes $O(mnc+nc^2)$ time; step 9 takes $O(mck)$ time. Overall, the running time is, asymptotically, dominated by the $O(mnc)$ term is step 8, with $c$ as in Eqn.~(\ref{eqn:cvalfinal}).

\subsection{References.} Our presentation in this chapter follows the derivations in~\cite{DrineasIKM16}. We also refer the interested reader to~\cite{MM2015,Wang2015} for related work.


\subsection*{Acknowledgements.} The authors would like to thank Ilse Ipsen for allowing them to use her slides for the introductory linear algebra lecture delivered at the PCMI Summer School. The first section of this chapter is heavily based on those slides. The authors would also like to thank Aritra Bose, Eugenia-Maria Kontopoulou, and Fred Roosta for their help in proofreading early drafts of this manuscript.

%
%
%
%
%
%
%


\end{document}

%% file: fig_matmul1_alg.tex
\begin{algorithm}[t] 
\begin{framed}

\textbf{Input:} $\bA \in \Rs{m \times n}$,
$\bB \in \mathbb{R}^{n \times p}$,
integer $c$ ($1 \le c \le n$), and
$\left\{ p_k \right\}_{k=1}^{n}$ s.t. $p_k \ge 0$ and $\sum_{k=1}^{n} p_k = 1$.

{\bf
Output:
}
$\bC \in \mathbb{R}^{m \times c}$ and
$\bR \in \mathbb{R}^{c \times p}$.
\vspace{0.1in}

\begin{enumerate}
\item For $t = 1$ to $c$,
\begin{itemize}
\item Pick $i_t \in \left\{1, \ldots, n \right\}$ with $\Probab{i_t = k} = p_k$, independently and with replacement.
\item Set $\bC_{*t} = \frac{1}{\sqrt{c p_{i_t}}}\bA_{*i_t} $ and $\bR_{t*} = \frac{1}{\sqrt{c p_{i_t}}}\bB_{i_t*}$.
\end{itemize}
\item Return $\bC\bR = \sum_{t=1}^{c} \frac{1}{cp_{i_t}} \bA_{*i_t} \bB_{i_t *}$.
\end{enumerate}
\end{framed}
\caption{ The \textsc{RandMatrixMultiply} algorithm }
\label{fig:BasicMatrixMultiplicationAlgorithm}
\end{algorithm}
%

%% file: alg_Sample_Fast.tex
\begin{algorithm}[t] 
\begin{framed}

\textbf{Input:} $\bA \in \Rs{n \times d}$, $\bvb \in
\Rs{n}$, and an error parameter $\epsilon \in (0,1)$.

\vspace{0.05in}

\textbf{Output:} $\tilde{\bx}_{opt} \in \Rs{d}$.

\begin{enumerate}

\item Let $r$ assume the value of Eqn.~(\ref{eqn:rvaluefinal}).

\item Let $\bS$ be an empty matrix.

\item \textbf{For} $t=1,\ldots,r$ (i.i.d. trials with replacement) \textbf{select uniformly at random} an integer from $\left\{1,2,\ldots,n\right\}$.

\begin{itemize}

\item \textbf{If} $i$ is selected, \textbf{then} append the column vector $\left(\sqrt{n/r}\right) \be_i$ to $\bS$, where $\be_i \in \Rs{n}$ is the $i$-th canonical vector.

\end{itemize}

\item Let $\bH \in \Rs{n\times n}$ be the normalized Hadamard transform
matrix.

\item Let $\bD \in \Rs{n \times n}$ be a diagonal matrix with
$$
\bD_{ii}
            = \left\{ \begin{array}{ll}
                         +1 & \mbox{, with probability $1/2$} \\
                         -1 & \mbox{, with probability $1/2$} \\
                      \end{array}
              \right.
$$
\item
Compute and return $\tilde{\bx}_{opt} = \left(\bS^T \bH \bD \bA \right)^\dagger \bS^T \bH \bD \bvb $.
\end{enumerate}

\end{framed}
\caption{The \textsc{RandLeastSquares} algorithm} \label{alg:alg_sample_fast}
\end{algorithm}

%% file: alg_RandLowRank.tex
\begin{algorithm}[t] 
\begin{framed}

\textbf{Input:} $\bA \in \Rs{m \times n}$, a rank parameter $k \ll \min\{m,n\}$, and an error parameter $\epsilon \in (0,1/2)$.

\vspace{0.05in}

\textbf{Output:} $\tilde{\bU}_{k} \in \Rs{m \times k}$.

\begin{enumerate}

\item Let $c$ assume the value of Eqn.~(\ref{eqn:cval5}).

\item Let $\bS$ be an empty matrix.

\item \textbf{For} $t=1,\ldots,c$ (i.i.d. trials with replacement) \textbf{select uniformly at random} an integer from $\left\{1,2,\ldots,n\right\}$.

\begin{itemize}

\item \textbf{If} $i$ is selected, \textbf{then} append the column vector $\left(\sqrt{n/c}\right) \be_i$ to $\bS$, where $\be_i \in \Rs{n}$ is the $i$-th canonical vector.

\end{itemize}

\item Let $\bH \in \Rs{n\times n}$ be the normalized Hadamard transform
matrix.

\item Let $\bD \in \Rs{n \times n}$ be a diagonal matrix with
$$
\bD_{ii}
            = \left\{ \begin{array}{ll}
                         +1 & \mbox{, with probability $1/2$} \\
                         -1 & \mbox{, with probability $1/2$} \\
                      \end{array}
              \right.
$$
\item
Compute $\bC = \bA \bD \bH \bS \in \mathbb{R}^{m \times c}$.
\item
Compute $\bU_C$, a basis for the column space of $\bC$.
\item
Compute $\bW = \bU_C^T\bA$ and (assuming that its rank is at least $k$), compute its top $k$ left singular vectors $\bU_{W,k}$.
\item
Return $\tilde{\bU}_k=\bU_C\bU_{W,k} \in \mathbb{R}^{m \times k}$.
\end{enumerate}

\end{framed}
\caption{The \textsc{RandLowRank} algorithm} \label{alg:alg_randlowrank}
\end{algorithm}

%% file: chapter.bbl
\begin{bibdiv}
\begin{biblist}


\bib{pcmi-chapter-randnla}{article}{
   author = {Petros Drineas and Michael W.  Mahoney},
   title={Lectures on Randomized Numerical Linear Algebra},
   conference={
      title={The Mathematics of Data},
   },
   book={
      series={IAS/Park City Math. Ser.},
      volume={25},
      publisher={Amer. Math. Soc., Providence, RI},
   },
   date={2018},
}

\bib{pcmi-chapter-martinsson}{article}{
   author={Martinsson, Per-Gunnar},
   title={Randomized methods for matrix computations},
   conference={
      title={The Mathematics of Data},
   },
   book={
      series={IAS/Park City Math. Ser.},
      volume={25},
      publisher={Amer. Math. Soc., Providence, RI},
   },
   date={2018},
}

\bib{pcmi-chapter-vershynin}{article}{
   author={Vershynin, Roman},
   title={Four lectures on probabilistic methods for data science},
   conference={
      title={The Mathematics of Data},
   },
   book={
      series={IAS/Park City Math. Ser.},
      volume={25},
      publisher={Amer. Math. Soc., Providence, RI},
   },
   date={2018},
}

\bib{pcmi-chapter-duchi}{article}{
   author={Duchi, John C.},
   title={Introductory lectures on stochastic optimization},
   conference={
      title={The Mathematics of Data},
   },
   book={
      series={IAS/Park City Math. Ser.},
      volume={25},
      publisher={Amer. Math. Soc., Providence, RI},
   },
   date={2018},
}

\bib{pcmi-chapter-maggioni}{article}{
   author={Maggioni, Mauro},
   title={Geometric and graph methods for high-dimensional data},
   conference={
      title={The Mathematics of Data},
   },
   book={
      series={IAS/Park City Math. Ser.},
      volume={25},
      publisher={Amer. Math. Soc., Providence, RI},
   },
   date={2018},
}

\bib{Drineas2016}{article}{
author = {P.~Drineas and M.~W.~Mahoney},
issn = {00010782},
journal = {Communications of the ACM},
month = {may},
number = {6},
pages = {80--90},
publisher = {ACM},
title = {{RandNLA: Randomized Numerical Linear Algebra}},
url = {http://dl.acm.org/ft{\_}gateway.cfm?id=2842602{\&}type=html},
volume = {59},
year = {2016}
}

\bib{Mah-mat-rev_BOOK}{book}{
  author =       {M.~W.~Mahoney},
  title =        {Randomized algorithms for matrices and data},
  publisher =    {NOW Publishers},
  year =         {2011},
  address =      {Boston},
  series =       {Foundations and Trends in Machine Learning},
}

\bib{GVL96}{book}{
  author =       {G.~H.~Golub and C.~F.~Van~Loan},
  title =        {Matrix Computations},
  publisher =    {Johns Hopkins University Press},
  year =         {1996},
  address =      {Baltimore},
}

\bib{Bjo15}{book}{
author= {{\AA.~Bj\"{o}rck}},
title = {Numerical Methods in Matrix Computations},
address = {Heidelberg},
publisher = {Springer},
year = {2015}
}

\bib{Strang88}{book}{
  author =       {G.~Strang},
  title =        {Linear Algebra and Its Applications},
  publisher =    {Harcourth Brace Jovanovich},
  year =         {1988},
}

\bib{TrefethenBau97}{book}{
  author =       {L.N.~Trefethen and D.~Bau~III},
  title =        {Numerical Linear Algebra},
  publisher =    {SIAM},
  year =         {1997},
  address =      {Philadelphia},
}

\bib{Stewart90}{book}{
  author =       {G. W. Stewart and J. G. Sun},
  title =        {Matrix Perturbation Theory},
  publisher =    {Academic Press},
  year =         {1990},
  address =      {New York},
}

\bib{Bhatia97}{book}{
  author =       {R. Bhatia},
  title =        {Matrix Analysis},
  publisher =    {Springer-Verlag},
  year =         {1997},
  address =      {New York},
}

\bib{MotwaniRaghavan95}{book}{
  author =       {R. Motwani and P. Raghavan},
  title =        {Randomized Algorithms},
  publisher =    {Cambridge University Press},
  year =         {1995},
  address =      {New York},
}

\bib{dkm_matrix1}{article}{
  author =       {P.~Drineas and R.~Kannan and M.~W.~Mahoney},
  title =        {{Fast {Monte Carlo} Algorithms for Matrices {I}: Approximating Matrix Multiplication}},
  journal =      {SIAM Journal on Computing},
  year =         {2006},
  volume =       {36},
  number =       {},
  pages =        {132--157},
}

\bib{Oli10}{article}{
  author =       {R.~I.~Oliveira},
  title =        {{Sums of random Hermitian matrices and an inequality by Rudelson}},
  journal =      {arXiv:1004.3821v1},
  year =         {2010},
}

\bib{Srivastava2010}{book}{
author = {N.~Srivastava},
title = {{Spectral sparsification and restricted invertibility}},
publisher = {Yale University},
address = {New Haven, CT}
url = {https://math.berkeley.edu/{~}nikhil/dissertation.pdf},
year = {2010}
}

\bib{Holodnak2015}{article}{
author = {J.~T.~Holodnak and I.~Ipsen},
journal = {SIAM J. Matrix Anal. Appl.},
number = {1},
pages = {110--137},
title = {{Randomized Approximation of the Gram Matrix: Exact Computation and Probabilistic Bounds}},
volume = {36},
year = {2015}
}

\bib{Ailon2009}{article}{
author = {N.~Ailon and B.~Chazelle},
issn = {0097-5397},
journal = {SIAM Journal on Computing},
number = {1},
pages = {302--322},
title = {{The Fast Johnson-Lindenstrauss Transform and Approximate Nearest Neighbors}},
volume = {39},
year = {2009}
}

\bib{AMT10}{article}{
  author =       {H.~Avron and P.~Maymounkov and S.~Toledo},
  title =        {Blendenpik: Supercharging {LAPACK}'s least-squares solver},
  journal =      {SIAM Journal on Scientific Computing},
  year =         {2010},
  volume =       {32},
  number =       {},
  pages =        {1217--1236},
}

\bib{DMMS11}{article}{
  author =       {P.~Drineas and M.~W.~Mahoney and S.~Muthukrishnan and T.~Sarl\'{o}s},
  title =        {Faster Least Squares Approximation},
  journal =      {Numerische Mathematik},
  year =         {2010},
  volume =       {117},
  number =       {2},
  pages =        {219--249},
}

\bib{Hoeffding1963}{article}{
author = {W.~Hoeffding},
journal = {Journal of the American Statistical Association},
number = {301},
pages = {13--30},
title = {{Probability inequalities for sums of bounded random variables}},
volume = {58},
year = {1963}
}

\bib{AL09}{article}{
 author = {N.~Ailon and E.~Liberty},
 title = {Fast Dimension Reduction Using Rademacher Series on Dual BCH Codes},
 journal = {Discrete Comput. Geom.},
 volume = {42},
 number = {4},
 year = {2009},
 issn = {0179-5376},
 pages = {615--630},
}
  
\bib{Woodruff2014}{article}{
url = {http://dx.doi.org/10.1561/0400000060},
year = {2014},
volume = {10},
journal = {Foundations and Trends in Theoretical Computer Science},
title = {Sketching as a Tool for Numerical Linear Algebra},
issn = {1551-305X},
number = {1-2},
pages = {1-157},
author = {D.~P.~Woodruff}
}

\bib{Saad2011}{book}{
author = {Y.~Saad},
publisher = {SIAM},
title = {{Numerical Methods for Large Eigenvalue Problems}},
year = {2011}
}

\bib{dkm_matrix2}{article}{
  author =       {P.~Drineas and R.~Kannan and M.~W.~Mahoney},
  title =        {Fast {Monte Carlo} Algorithms for Matrices {II}:Computing a Low-Rank Approximation to a Matrix},
  journal =      {SIAM Journal on Computing},
  year =         {2006},
  volume =       {36},
  number =       {},
  pages =        {158--183},
}

\bib{HMT09_SIREV}{article}{
  author =       {N.~Halko and P.-G.~Martinsson and J.~A.~Tropp},
  title =        {Finding structure with randomness: Probabilistic algorithms for constructing approximate matrix decompositions},
  journal =      {SIAM Review},
  year =         {2011},
  volume =       {53},
  number =       {2},
  pages =        {217--288},
}

\bib{MD2016}{article}{
author = {M.~W.~Mahoney and P.~Drineas},
title = {{Structural Properties Underlying High-quality Randomized Numerical Linear Algebra Algorithms}},
journal = {CRC Handbook on Big Data},
pages = {137-154},
year = {2016},
}

\bib{MM2015}{inproceedings}{
author = {C.~Musco and C.~Musco},
booktitle = {Neural Information Processing Systems (NIPS)},
eprint = {arXiv:1504.05477},
title = {{Stronger and Faster Approximate Singular Value Decomposition via the Block Lanczos Method}},
year = {2015}
}

\bib{DrineasIKM16}{article}{
  author    = {P.~Drineas and I.~Ipsen and E.~Kontopoulou and M.~Magdon{-}Ismail},
  title     = {{Structural Convergence Results for Low-Rank Approximations from Block Krylov Spaces}},
  journal   = {arXiv:1609.00671},
  year      = {2016},
}

\bib{BDM2014}{article}{
author = {C.~Boutsidis and P.~Drineas and M.~Magdon-Ismail},
title = {Near-Optimal Column-Based Matrix Reconstruction},
journal = {SIAM Journal on Computing},
volume = {43},
number = {2},
pages = {687-717},
year = {2014},
}

\bib{Wang2015}{article}{
archivePrefix = {arXiv},
arxivId = {1508.06429},
author = {S.~Wang and Z.~Zhang and T.~Zhang},
eprint = {1508.06429},
journal = {arXiv:1508.06429v2 [cs.NA]},
pages = {1--22},
title = {{Improved Analyses of the Randomized Power Method and Block Lanczos Method}},
url = {http://arxiv.org/abs/1508.06429},
year = {2015}
}


\end{biblist}
\end{bibdiv}
